\documentclass[final]{amsart}

\usepackage{amssymb}
\usepackage{amsthm}

\usepackage{graphicx}
\usepackage{color}

\makeatletter
\@addtoreset{equation}{section}

\makeatother

\newtheorem{theorem}{Theorem}[section]
\newtheorem{proposition}[theorem]{Proposition}
\newtheorem{corollary}[theorem]{Corollary}
\newtheorem{lemma}[theorem]{Lemma}
\theoremstyle{definition}
\newtheorem{definition}[theorem]{Definition}
\newtheorem{remark}[theorem]{Remark}

\begin{document}
	
	\title[]{Overhanging of membranes and filaments adhering to periodic graph substrates}
	\author{Tatsuya Miura}
	\address{Graduate School of Mathematical Sciences, University of Tokyo, 3-8-1 Komaba, Meguro, Tokyo, 153-8914 Japan}
	\email{miura@ms.u-tokyo.ac.jp}
	%\subjclass[2010]{XXXXX, and XXXXX}

\begin{abstract}
This paper mathematically studies membranes and filaments adhering to periodic patterned substrates in a one-dimensional model.
The problem is formulated by the minimizing problem of an elastic energy with a contact potential on graph substrates.
Global minimizers (ground states) are mainly considered in view of their graph representations.
Our main results exhibit sufficient conditions for the graph representation and examples of situations where any global minimizer must overhang.
\end{abstract}

\maketitle

\section{Introduction}

The figuration of elastic bodies is complicated to comprehend, in particular, if external factors and constraints are taken into consideration.
This paper is devoted to a theoretical study of slender elastic bodies adhering to solid substrates.

The contact and adhesion problems between soft objects and solid substrates appear in various contexts.
For example, complex adhesion patterns are observed when soft nano-objects, as graphene \cite{NiMaPoPLSM11,YaPLHuFuEiCu12} or carbon nanotubes \cite{DeMaApAv01}, are sheeted on rough patterned substrates.
The adhesion property is also known for vesicles (cf.~\cite{SeLi90}).
More broadly, in contact mechanics \cite{Jo87}, it is a central question to ask how elastic bodies contact rough substrates \cite{PeTo75,YaAnMo15}.
This question is relevant for many motivating problems as rubber friction \cite{Pe01} or adhesion in biological systems as geckos \cite{HuGoHoSpAr07,PeGo03,PuLe08}.
Recently, there are remarkable progresses in ``elasto-capillary'' problems \cite{RoBi10}.
The elasto-capillary problems essentially relate to our problem in the sense that they are focused on the competition between elasticity and adhesiveness.

\subsection{Our model}

This paper mathematically studies the adhesion problems of filaments and membranes in a one-dimensional setting, as in \cite{PL08}.
To be more precise, we consider the minimizing problem of the energy
\begin{eqnarray}\label{energy0}
\mathcal{E}[\gamma]=\int_{\gamma}ds\left[\frac{C}{2}\kappa^2(s) + \sigma\left(\gamma(s)\right)\right]
\end{eqnarray}
defined for planar curves $\gamma$.
Here $\kappa$ and $s$ denote the curvature and arc length parameter, respectively.
Admissible curves $\gamma$ (corresponding to elastic bodies) are constrained in the upper side of a given $\lambda$-periodic substrate function $\psi_\lambda$ as in Figure \ref{figsubstrate}.
The constant $C>0$ corresponds to the bending rigidity.
The contact potential $\sigma$ is defined as $\sigma=\sigma_F$ in the free part and $\sigma=\sigma_B$ in the bounded part, where $0<\sigma_B<\sigma_F$ are constants.
The constants $\sigma_B$ and $\sigma_F$ correspond to tension or surface energies.
(See Section \ref{sectprelimi} for details.)

Our energy is a simple generalization of the modified total squared curvature, so-called Euler's elastica energy (see \cite{BeMu04,BrPoWi15,Br92,Li98,Li98b} and also \cite{Lo44,Sa08,Si08}), so that an adhesion effect (contact potential) is included.
Its minimization invokes a free boundary problem of the elastica equation, i.e., the free part of any minimizing curve satisfies the curvature equation $C(\kappa_{ss}+\kappa^3/2)-\sigma_F\kappa=0$.
The free boundary conditions are concerned with curvature jumps (see \cite{PL08} and also \cite{LaLi86,OlFr05,RoBi10,SeLi90}).
Our model can be regarded as an elastic version of wetting problems (cf. \cite{dGBWQu04,Me12}).

Our model concerns only the bending modes of filaments or membranes and neglects the stretching modes.
As mentioned in \cite{PL08}, the underlying physical assumptions are that elastic bodies are sufficiently thin, vary only in one direction, and move along substrates freely (no friction).
The stretching modes should be taken into account in fully two-dimensional models, even for thin films without friction (see e.g.\ \cite{HuRoBi11,YaPLHuFuEiCu12}).

\begin{figure}[tb]
	\begin{center}
		\begin{tabular}{cc}
			\begin{minipage}[b]{0.55\hsize}
				\begin{center}
					\includegraphics[width=40mm]{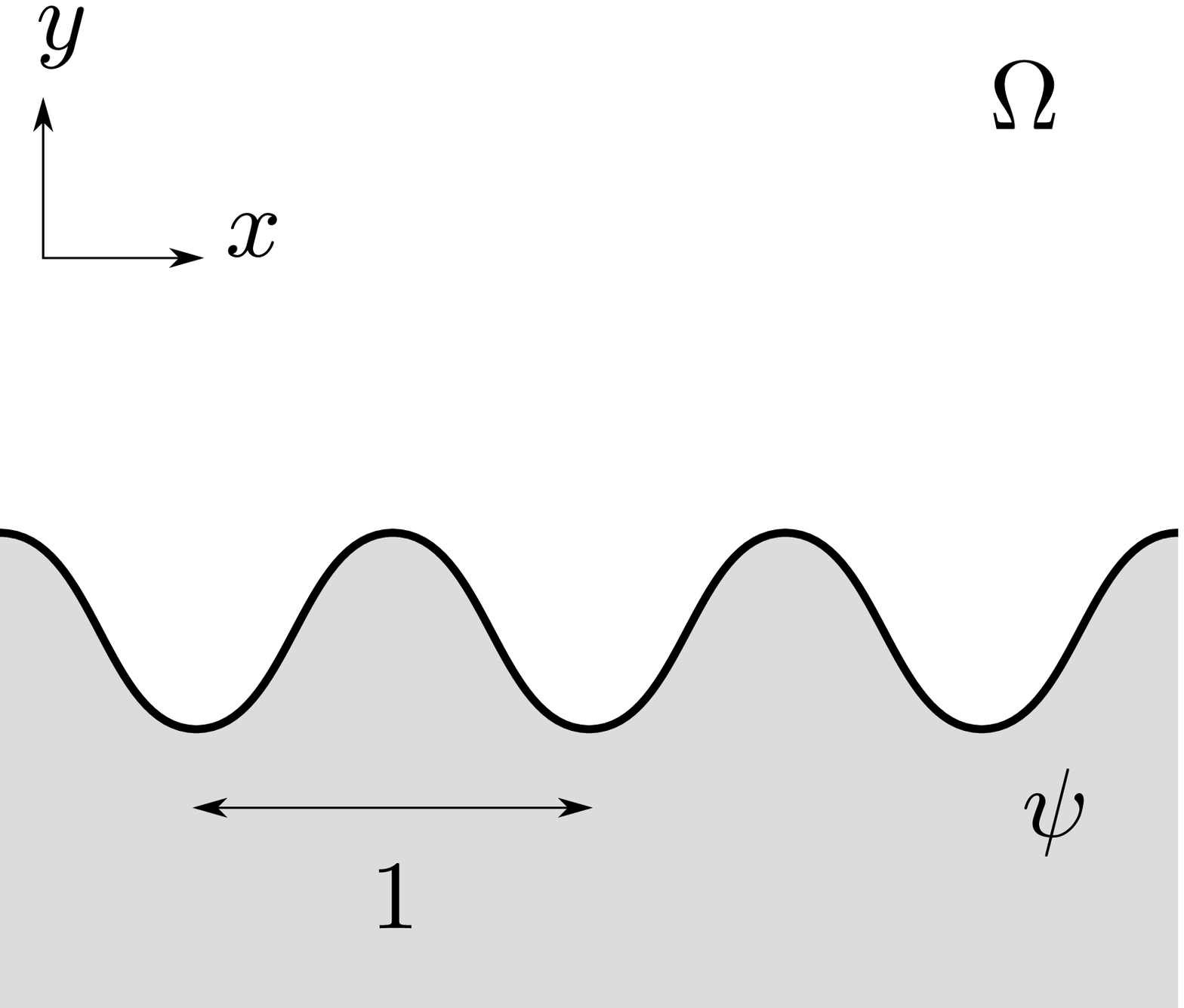}
				\end{center}
			\end{minipage}
			\begin{minipage}[b]{0.40\hsize}
				\begin{center}
					\includegraphics[width=27mm]{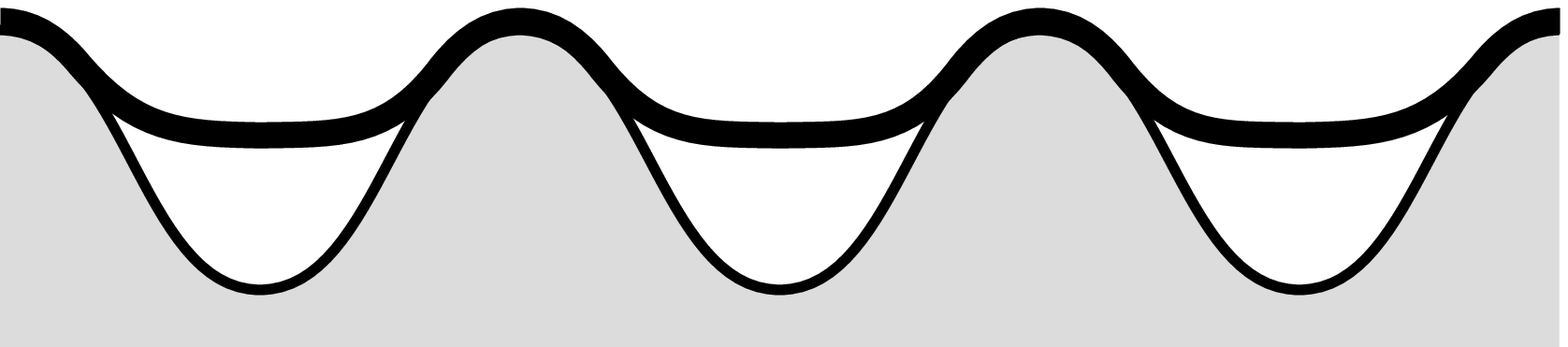}
				\end{center}
				\begin{center}
					\includegraphics[width=27mm]{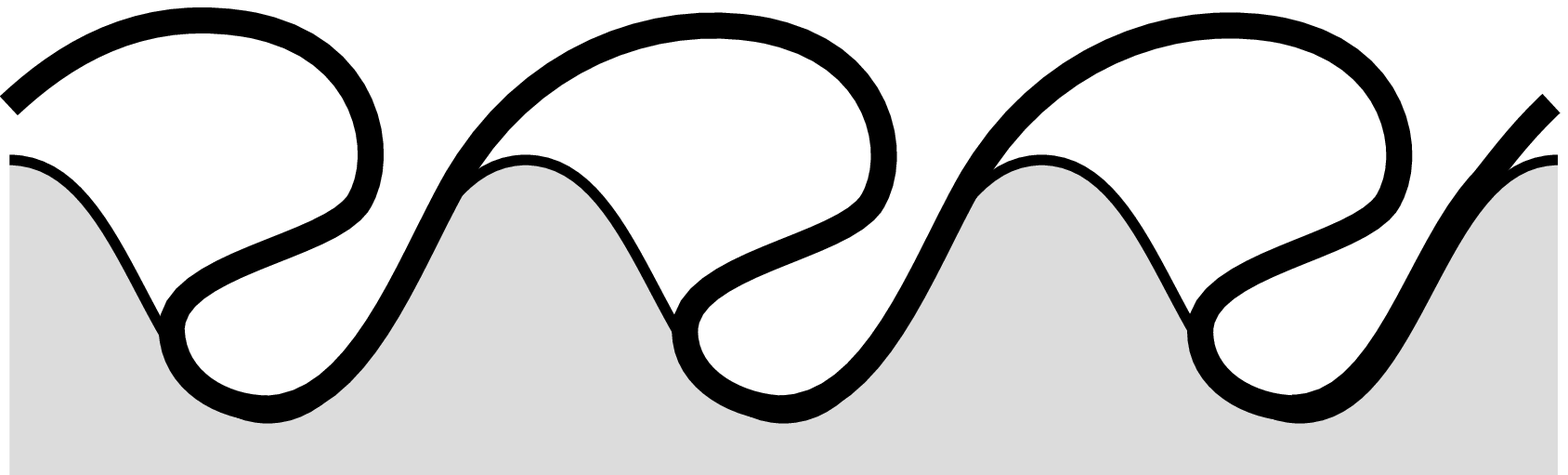}
				\end{center}
				\begin{center}
					\includegraphics[width=27mm]{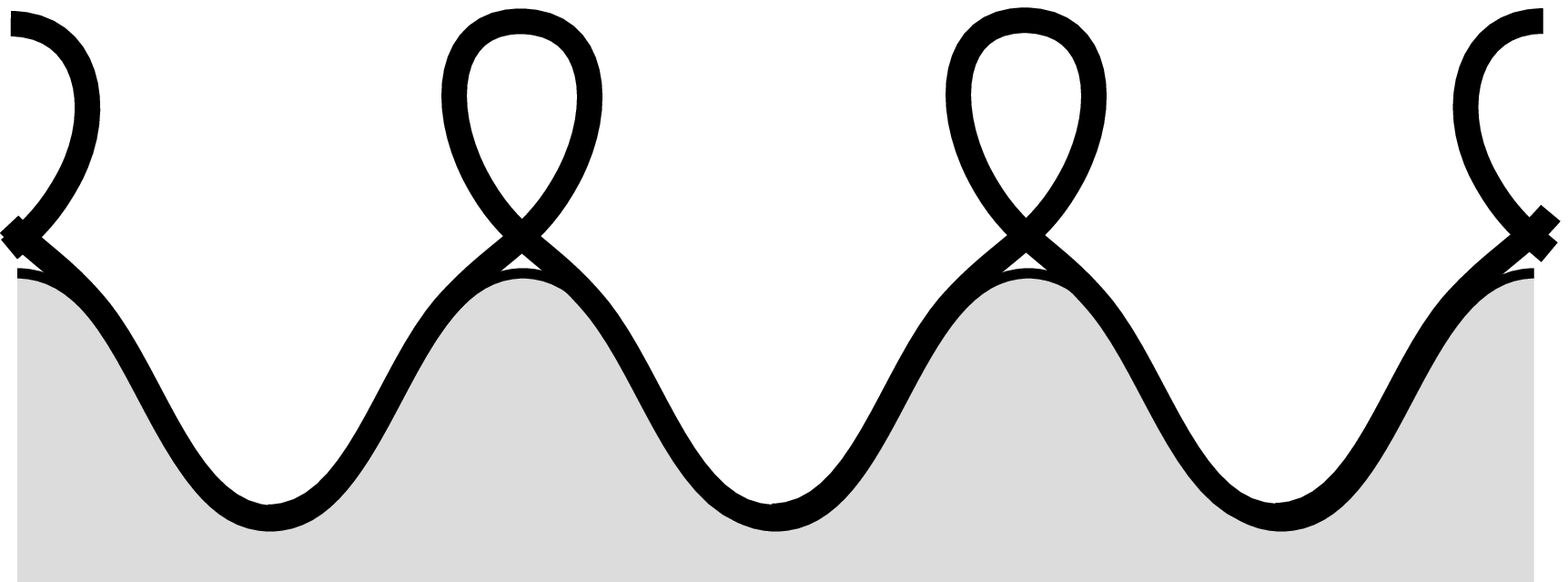}
				\end{center}
			\end{minipage}
		\end{tabular}
		\caption{Periodic substrate function $\psi$ and periodic admissible curves.
			Admissible curves may overhang or self-intersect.}
		\label{figsubstrate}
	\end{center}
\end{figure}

\subsection{Our goal}

The local laws (as the elastica equation or boundary conditions) are well-known in our model since similar models have been widely studied (e.g.\ in \cite{LaLi86,OlFr05,RoBi10,SeLi90}).
The fundamental goal of this paper is to know the whole shapes of minimizers in our model.
However, it is not realistic to determine the exact whole shapes of minimizers for arbitrary parameters and a substrate.
This paper focuses on whether minimizers are represented by the graphs of functions or not.

Whether minimizers are graphs or have overhangs is an important consequence on the shapes.
In fact, the absence of overhangs guarantees that the shape of a solution is not so ``complex'', in particular, there is no self-intersection.
Conversely, the presence of overhangs implies the possibility of self-intersections.
If once membranes or filaments self-intersect, then other mechanisms (not taken into account in our model) may yield more complex shapes as rackets \cite{CoMa03,CrSeBu09,ZhHuLiHwZuGa07} (see also \cite{RoBi10}).

An a priori guarantee of the graph representation is also important for the theoretical study.
Such a guarantee rigorously justifies the graph setting, i.e., the assumption to consider only graph curves as admissible curves.
The graph setting yields strong topological and morphological constraints, and hence makes the analysis considerably simpler.
In fact, there are theoretical studies \cite{Ke16,Mi16,PL08} concerning the whole shapes of minimizers in our model, but all of them rely on the graph setting.
The paper \cite{PL08} particularly depends on the graph setting since its analysis crucially uses the small slope approximation.

\subsection{Main results}

The present paper gives the first rigorous study on the graph representations of global minimizers (ground states).
A theoretical reason to consider only global minimizers is that the shapes of local minimizers (metastable states) may be more complicated even for parameters ensuring the graph representations of global minimizers (see Section \ref{sectdiscuss} for details).
The assumption of global minimality would be however appropriate for some experimental situations, for example, thin films on substrates with wetting fluids at the interfaces (almost no friction) as in \cite{HuRoBi11}.
In addition, as a mathematical assumption, the present paper assumes that curves $\gamma$ and a substrate $\psi_\lambda$ have a same period $\lambda$.

To describe our results, it is convenient to recall the typical length scale $\ell=\sqrt{C/\sigma_F}$, which compares bending rigidity and surface tension.
The scale $\ell$ is called the elasto-capillary length e.g.\ in \cite{HuRoBi11,RoBi10}.
As mentioned in \cite{HuRoBi11,RoBi10}, the scale $\ell$ appears as a typical bending scale of an elastic body.
We also use the length scale $r=\|\psi_\lambda''\|_\infty^{-1}$ which is the reciprocal of the maximum of the second derivative.
The scale $r$ roughly corresponds to the minimal bending scale of $\psi_\lambda$.
Moreover, the dimensionless ratio $\alpha=\sigma_B/\sigma_F$ is also important since it corresponds to adhesiveness.

Global minimizers are flat in many limiting cases; dominant bending effect ($C=\infty$), no adhesion ($\sigma_B=\sigma_F$) or flat substrate ($\psi=0$).
Hence, the graph representation is expected at least nearly the above cases.
Indeed, Theorem \ref{thmnooverhang} and Theorem \ref{thmnooverhang2} give explicit conditions ensuring that global minimizers must be graphs.
The first condition is described as $\alpha^{-1}-1\ll (\ell/\lambda)^2$.
In particular, this condition is satisfied as the limits $C\to\infty$ and $\sigma_B\to\sigma_F$.
The second condition is described as $(r/\lambda)^2\gg \alpha^{-1}+(\ell/\lambda)^{-2}$.
In particular, this condition is satisfied as the limit $r\to\infty$, which means a second order flatness of $\psi_\lambda$.
Our proof uses only energy arguments; we compare the energies of all non-graph curves and special graph competitors.

On the other hand, even if $\psi$ is smooth of class $C^\infty$, it turns out that there are situations such that global minimizers are overhanging, i.e., not represented by graphs.
The mechanism of overhangs is involved, so we deal with only special substrates like ``fakir carpets'' (see the figures in Section \ref{sectoverhang}).
Our result indicates that the wave height length scale $H$ and dimensionless ``deviation'' $\Delta:=\min\{\lambda,H\}/(\lambda+2H)$ of a fakir carpet appear as characteristic quantities.
More precisely, as a main result (Theorem \ref{thmoverhanging}), we rigorously prove that global minimizers must overhang if $\psi_\lambda$ is smooth but shaped like a fakir carpet and moreover the relations $r\ll\ell\ll\min\{\lambda,H\}$ and $\alpha\ll\Delta$ are satisfied.
Our proof is based on a geometric viewpoint to classify possible global states of non-overhanging curves, and an energy estimate for each of the cases.
A special overhanging competitor is then constructed in view of the optimal bending scale $\ell$.
We notice that the condition $r\ll \ell$ requires that $\ell$ is not arbitrary small for overhangs.
However, we also prove that if such substrates are Lipschitz (i.e., folding singularly $r=0$), then $\ell$ can be arbitrary small for a fixed substrate (Theorem \ref{thmoverhanging2}).
To this end, we need further discussion for local bending structure, but we still use only energy arguments.

\subsection{Related mathematical results}

In the rest of this section, reviewing related mathematical literature, we see that in our one-dimensional problem both the contact potential and the total squared curvature play crucial roles for overhangs.

There is much mathematical literature of first order energies with contact potentials on flat substrates (see e.g.\ \cite{AlCa81,CaFr85,Ya92,Ya94} for graphs, \cite{CaFr85,Me12} for the boundary of sets, and references therein).
The problems in the cited papers roughly correspond to our problem with $C=0$ and $\psi_\lambda\equiv0$ (but in higher dimensions).
In first order cases, solutions may have edge singularities at the free boundary and the contact angle $\theta$ satisfies Young's equation $\cos\theta=\sigma_B/\sigma_F$.
In higher dimensional cases, the contact potential may imply the loss of graph representation even in first order cases (cf.\ \cite{Ya92}).
However, although our substrates are not flat, our problem is one-dimensional and periodic, so the graph setting would be still suitable while $C=0$.

To our knowledge, there is little mathematical literature of higher order problems with contact potentials except the aforementioned papers \cite{Ke16,Mi16}.
In \cite{Mi16}, the author obtains an energy expansion as $C\to0$.
The paper \cite{Ke16} studies a discretization of our model and proposes numerical results.
As already mentioned, the papers \cite{Ke16,Mi16} assume that admissible curves are graphs.

The total squared curvature is higher order and a main reason of the loss of graph representation.
In fact, it is well-known that there are non-graph solutions to the elastica equation, which our minimizers obey in the free part (see e.g.\ the figures in \cite{Br92,Lo44}).
Thus, if we impose suitable fixed boundary conditions, it is not difficult to prove that a global minimizer of the modified total squared curvature overhangs.
Our problem is a free boundary problem, and hence the graph representation problem is more involved.

We finally mention that, in dynamical problems of curves (without substrates), the graph representations of solutions have also been concerned.
Although the $L^2$-gradient flow of the length energy (curve shortening flow) preserves the graph property \cite{EcHu91}, one of the modified total squared curvature (curve shortening-straightening flow \cite{Li98b,NoOk14,Ok11,Po96}) may lose in the middle even in the periodic setting \cite{ElMP01}.
However, in such a periodic case (without external factors), stationary global minimizers are only straight lines.
Our problem takes an adhesion effect into account and thus even global minimizers may not be graphs.

\subsection{Organization}

This paper is organized as follows.
Basic notation and definitions are prepared in Section \ref{sectprelimi}.
Section \ref{sectgraph} provides some sufficient conditions for the graph representations of global minimizers.
In Section \ref{sectoverhang}, we prove the existence of situations where global minimizers must overhang.
In Section \ref{sectdiscuss}, we give further discussion on overhangs and also mention self-intersections and local minimizers.
Section \ref{sectconclusion} is the conclusion section.

\section{Preliminaries: curves, energy and quantities}\label{sectprelimi}

In this section we prepare notation of admissible curves and the total energy and then formulate our problem.
For simplicity, throughout this paper, we impose normalizations with respect to the wavelength and tension.
In the last subsection we mention the relation between our normalized problem and original physical quantities.

\subsection{Definition of admissible curves}

Let $\psi:\mathbb{R}\to\mathbb{R}$ be a continuous function with $1$-periodicity, that is, $\psi(x)=\psi(x+1)$ for any $x\in\mathbb{R}$.
Let $\Omega\subset\mathbb{R}^2$ be the strict epigraph of $\psi$:
\begin{eqnarray}
\Omega:=\left\{(x,y)\in\mathbb{R}^2 \mid y>\psi(x)\right\}. \nonumber
\end{eqnarray}
Denote its closure by $\overline{\Omega}=\{y\geq\psi(x)\}$ and the boundary by $\partial\Omega=\{y=\psi(x)\}$.
The set $\overline{\Omega}$ corresponds to the upper side of the substrate $\psi$.

Let $I$ be the open interval $(0,1)$.
We denote by $H^2_{\Omega}$ the set of all curves $\gamma\in H^2(I;\mathbb{R}^2)$ such that $\gamma$ is regular and confined in $\overline{\Omega}$, that is, $|\dot{\gamma}(t)|>0$ and $\gamma(t)\in\overline{\Omega}$ for any $t\in\bar{I}$.
The $H^2$-Sobolev setting corresponds to the square integrability of curvature.
Recall that any regular $H^2$ curve is a regular $C^1$ curve by Sobolev embedding and hence $\gamma$ and $\dot{\gamma}$ are defined pointwise in $\bar{I}$ (including the endpoints).

Moreover we say that a curve $\gamma=(x,y)\in H^2_\Omega$ is admissible if it satisfies the following periodic boundary condition:
\begin{eqnarray}
x(0)=0,\ x(1)=1,\ y(0)=y(1),\ \dot{\gamma}(0)=\dot{\gamma}(1). \nonumber
\end{eqnarray}
We denote the set of admissible curves by $\mathcal{A}\subset H^2_\Omega$.
Remark that the set $\mathcal{A}$ consists of the restrictions to $I$ of regular curves $\gamma=(x,y)\in H^2_{loc}(\mathbb{R};\overline{\Omega})$ such that $x(0)=0$ and $\gamma(t+1)=\gamma(t)+(1,0)\in\mathbb{R}^2$ for any $t\in\mathbb{R}$.

In this setting, admissible curves may have any self-intersections, and thus it would not be compatible with membrane problems.
However, we easily confirm that all the results in this paper are valid in a membrane setting (see Section \ref{sectdiscuss} for details).

\subsection{Definition of the total energy}

For any admissible curve $\gamma\in\mathcal{A}$ we define the length of one period by
\begin{equation}
L_\gamma:=\int_I |\dot{\gamma}(t)|dt\geq1.
\end{equation}
The lower bound follows by the periodicity, in particular, by $x(0)=0$ and $x(1)=1$.

Using the arc length parameterization $0\leq s\leq L_\gamma$, for given constants $\varepsilon>0$ and $0<\alpha<1$, we define the total energy of one period:
\begin{equation}\label{energy1}
E[\gamma]:=\int_0^{L_\gamma} \left[\varepsilon^2|\gamma_{ss}(s)|^2+\Theta(\gamma(s))\right] ds,
\end{equation}
where $\Theta:\overline{\Omega}\to\mathbb{R}$ is defined as $\Theta\equiv1$ in $\Omega$ and $\Theta\equiv\alpha$ on $\partial\Omega$.
Note that $|\gamma_{ss}|^2$ is equal to the squared curvature $\kappa^2$.

Then our problem is formulated as
\begin{eqnarray}\label{minimizingproblem}
\min_{\gamma\in\mathcal{A}}E[\gamma].
\end{eqnarray}
Our purpose is to know the shapes of global minimizers, i.e., curves $\gamma\in\mathcal{A}$ satisfying $E[\gamma]=\min_\mathcal{A}E$.
Hereafter a global minimizer is often called a minimizer simply.

The problem (\ref{minimizingproblem}) is determined by the quantities $\varepsilon>0$ and $0<\alpha<1$ and the substrate $\psi$.
The quantity $\varepsilon$ corresponds to the normalized elasto-capillary length scale (or bending scale) of minimizing curves.
The coefficient $\alpha$ corresponds to adhesiveness.
The smaller $\alpha$ is, the easier the curves become to adhere.
From Section \ref{sectgraph}, changing the parameters $\varepsilon$, $\alpha$ and $\psi$, we consider whether minimizers are represented as graphs or not.

\subsection{Existence of minimizers}

We shall state the existence of solutions to the problem (\ref{minimizingproblem}).

\begin{theorem}\label{thmexistence}
	The problem (\ref{minimizingproblem}) admits a minimizer.
\end{theorem}

This is proved by a basic direct method in the calculus of variations.
However, we need some careful arguments to prove it rigorously and thus postpone the precise proof until \ref{appendix1}.
In general, the uniqueness is not expected in this problem.

\subsection{Bounds for minimum}

We have simple bounds for the minimum of $E$:
\begin{eqnarray}\label{energybound}
\alpha \leq \min_{\mathcal{A}} E \leq 1.
\end{eqnarray}
The upper bound follows since the trivial straight line $\bar{\gamma}(t)=(t,c)$, where $c$ is a constant larger than the maximum of $\psi$, belongs to $\mathcal{A}$ and satisfies $E[\bar{\gamma}]=1$.
The lower bound follows since for any $\gamma\in\mathcal{A}$ we have $L_\gamma\geq1$ and also
\begin{eqnarray}\label{energybound2}
E[\gamma] \geq \int_{\gamma}\alpha ds=\alpha L_\gamma.
\end{eqnarray}

\subsection{On normalization of the problem}

As mentioned, our problem is normalized with respect to the wavelength and tension.
To clarify the normalizations, we consider the relation between our normalized problem and the original one.
In the original problem, we only assume that admissible curves and a substrate have a same periodicity $\lambda>0$.

If the substrate $\psi_\lambda$ and admissible curves $\gamma_\lambda$ have a general period $\lambda>0$ and tension in the energy is also a general positive number, then as in Introduction our original problem is the minimization of
\begin{equation}\label{energy2}
\mathcal{E}[\gamma_\lambda]:=\int_0^{L_{\gamma_\lambda}} \left[\frac{C}{2}|(\gamma_\lambda)_{ss}(s)|^2+\sigma(\gamma_\lambda(s))\right] ds,
\end{equation}
where $C>0$ and $\sigma$ is defined as $\sigma\equiv\sigma_F$ in $\{y>\psi_\lambda(x)\}$ and $\sigma\equiv\sigma_B$ on $\{y=\psi_\lambda(x)\}$ with $0<\sigma_B<\sigma_F$.

Normalizing $\psi_\lambda$ and $\gamma_\lambda$ by rescaling as $\psi(x)=\lambda^{-1}\psi_\lambda(\lambda x)$ and $\gamma(s)=\lambda^{-1}\gamma_\lambda(\lambda s)$, we have
\begin{eqnarray}
\mathcal{E}[\gamma_\lambda]=\lambda\sigma_F E[\gamma], \nonumber
\end{eqnarray}
where the dimensionless quantities $\varepsilon>0$ and $0<\alpha<1$ in $E$ are defined as
\begin{eqnarray}\label{quantities}
\varepsilon:=\frac{1}{\lambda}\sqrt{\frac{C}{2\sigma_F}},\quad \alpha:=\frac{\sigma_B}{\sigma_F}.
\end{eqnarray}
Since we only used a similarity transformation, the shapes of curves and a substrate are maintained.
Thus, the minimizing problem of (\ref{energy2}) is equivalent to (\ref{minimizingproblem}) up to rescaling.
We finally recall that $\varepsilon$ has the same scale as the normalized elasto-capillary length scale, i.e., $\varepsilon\sim\ell/\lambda.$

\section{Graph solutions}\label{sectgraph}

In this section we prove that, under suitable conditions for $\varepsilon$, $\alpha$, and $\psi$, the problem (\ref{minimizingproblem}) admits only graph minimizers.

We shall give the definition of graph curves.

\begin{definition}[Graph curves]
	We say that $\gamma=(x,y)\in\mathcal{A}$ is a graph curve if $x'(t)>0$ for any $t\in\bar{I}$.
\end{definition}

\begin{remark}
	By the periodicity, the condition $x'(t)>0$ is equivalent to $x'(t)\neq0$.
	Any graph curve $\gamma$ is represented by an $H^2$ function in the $y$-direction; more precisely, there exists a $1$-periodic function $u\in H^2_{loc}(\mathbb{R})$ such that its graph curve $(\cdot,u(\cdot))\in\mathcal{A}$ is a reparameterization of $\gamma$.
\end{remark}

\subsection{Statements and discussion}

We first observe the following limiting cases; $\varepsilon=\infty$, $\alpha=1$ and $\psi\equiv0$.
We easily notice that in all the cases minimizers are only straight lines.
Indeed, in the case that $\varepsilon=\infty$ or $\alpha=1$, our energy is regarded as the (modified) total squared curvature, which admits only straight line minimizers under the periodicity.
Moreover, in the case that $\psi\equiv0$, it is trivial that a unique minimizer is the completely adhering straight line.

By the above observation, when $\varepsilon\gg1$, $\alpha\approx1$ or $\psi\approx0$, we expect that any minimizer is nearly flat and hence a graph curve.
In fact, the following two statements hold.

\begin{theorem}\label{thmnooverhang}
	Suppose that $(\pi^2\varepsilon^2+1)\alpha\geq1$.
	Then, independently of $\psi$, any minimizer of $(\ref{minimizingproblem})$ is a graph curve.
\end{theorem}

\begin{theorem}\label{thmnooverhang2}
	Suppose that $\psi$ is of class $C^2$ and $\|\psi''\|_\infty^2\leq\frac{8\pi^2}{8/\alpha+1/\varepsilon^2}$.
	Then any minimizer of (\ref{minimizingproblem}) is a graph curve.
\end{theorem}

\begin{remark}
	Theorem \ref{thmnooverhang} immediately implies that, if we fix $\varepsilon$ and take $\alpha\approx1$, or fix $\alpha$ and take $\varepsilon\gg1$, then any minimizer is a graph curve.
	
	In view of the original physical quantities (\ref{quantities}), the condition that $(\pi^2\varepsilon^2+1)\alpha\geq1$ in Theorem \ref{thmnooverhang} is read as
	\begin{eqnarray}
	\frac{\pi^2}{2}\frac{C}{\lambda^2\sigma_F}=\frac{\pi^2}{2}\left(\frac{\ell}{\lambda}\right)^2 \geq\frac{\sigma_F}{\sigma_B}-1. \nonumber
	\end{eqnarray}
	This is enough to indicate the following two qualitative features: if the effect of adhesion is weak ($\sigma_B\to\sigma_F$), or the effect of bending is strong ($C\to\infty$), then any minimizer must be a graph curve.
\end{remark}

\begin{remark}
	Theorem \ref{thmnooverhang2} states that, for any $\varepsilon$ and $\alpha$ which may be small, if the substrate $\psi$ is sufficiently flat in a second order sense $\psi''\approx0$, then our problem still admits only graph curve minimizers.
	
	Recall that the sup norm $\|\psi''\|_\infty=\max_{x\in\mathbb{R}}|\psi''(x)|$ is also a dimensionless quantity since $\|\psi''\|_\infty=\lambda/r$, where $r=\|\psi_\lambda''\|_\infty^{-1}$ corresponds to the minimal bending scale of the original substrate $\psi_\lambda$.
	By (\ref{quantities}), the condition in Theorem \ref{thmnooverhang2} can be also expressed by the original quantities.
\end{remark}

\subsection{Proof of graph representation}

In this section we prove Theorem \ref{thmnooverhang} and Theorem \ref{thmnooverhang2}.
We first obtain a lower bound for the energies of non-graph curves.
This is a key step to prove our theorems.

\begin{proposition}\label{propnooverhang}
	Any non-graph curve $\gamma\in\mathcal{A}$ satisfies
	\begin{eqnarray}
	E[\gamma] > \min\left\{1,(\pi^2\varepsilon^2+1)\alpha\right\}. \nonumber
	\end{eqnarray}
\end{proposition}

\begin{proof}
	By (\ref{energybound2}), any curve $\gamma\in\mathcal{A}$ with $L_\gamma>1/\alpha$ satisfies $E[\gamma]\geq\alpha L_\gamma>1$.
	Thus it suffices to prove that any non-graph curve $\gamma\in\mathcal{A}$ with $1\leq L_\gamma\leq 1/\alpha$ satisfies $E[\gamma]>(\pi^2\varepsilon^2+1)\alpha$.
	By the Cauchy-Schwarz inequality, we have
	\begin{eqnarray}
	E[\gamma] \geq \varepsilon^2\int_{\gamma}\kappa^2ds +\alpha L_\gamma \geq \frac{\varepsilon^2}{L_\gamma}\left(\int_\gamma|\kappa| ds \right)^2 +\alpha L_\gamma. \nonumber
	\end{eqnarray}
	Moreover, since $\gamma$ is non-graph and periodic, its total absolute curvature (i.e., the total variation of its tangential angle) is larger than $\pi$.
	Noting that $1\leq L_\gamma\leq 1/\alpha$, we have
	\begin{eqnarray}
	E[\gamma] > \frac{\varepsilon}{1/\alpha}\pi^2 +\alpha=(\pi^2\varepsilon^2+1)\alpha. \nonumber
	\end{eqnarray}
	The proof is complete.
\end{proof}

We are in position to prove the theorems.

\begin{proof}[Proof of Theorem \ref{thmnooverhang}]
	Proposition \ref{propnooverhang} and the assumption $(\pi^2\varepsilon^2+1)\alpha\geq1$ imply that any non-graph curve $\gamma\in\mathcal{A}$ satisfies $E[\gamma]>1$.
	By the upper bound in (\ref{energybound}), such a curve is not a minimizer.
\end{proof}

\begin{proof}[Proof of Theorem \ref{thmnooverhang2}]
	By Theorem \ref{thmnooverhang}, we may assume $(\pi^2\varepsilon^2+1)\alpha<1$.
	Thus, by Proposition \ref{propnooverhang}, we only need to prove that any $\gamma\in\mathcal{A}$ such that $E[\gamma]>(\pi^2\varepsilon^2+1)\alpha$ is not a minimizer.
	
	We compare such $\gamma$ with the completely adhering competitor $\tilde{\gamma}:=(\cdot,\psi(\cdot))\in\mathcal{A}$, which is a graph curve since $\psi$ is of class $C^2$.
	Noting that $2\|\psi'\|_\infty\leq \|\psi''\|_\infty$ by the $1$-periodicity, we find that the curve $\tilde{\gamma}$ satisfies
	\begin{eqnarray}
	E[\tilde{\gamma}] &=& \varepsilon^2\int_I \frac{|\psi''|^2}{(1+|\psi'|^2)^{5/2}} + \alpha\int_I \sqrt{1+|\psi'|^2} \nonumber\\
	&\leq& \varepsilon^2\|\psi''\|_\infty^2 + \alpha\left(1+\frac{1}{2}\left(\frac{\|\psi''\|_\infty}{2}\right)^2\right). \nonumber\\
	&=& \left(\frac{8\varepsilon^2+\alpha}{8\alpha}\|\psi''\|_\infty^2+1\right)\alpha. \nonumber
	\end{eqnarray}
	The assumption on $\|\psi''\|_\infty$ immediately implies that $\frac{8\varepsilon^2+\alpha}{8\alpha}\|\psi''\|_\infty^2 \leq \pi^2\varepsilon^2$ and hence $E[\gamma]>E[\tilde{\gamma}]$.
	Therefore, the curve $\gamma$ does not minimize $E$.
\end{proof}

\begin{remark}
	In the proof of Proposition \ref{propnooverhang}, using the inequality of arithmetic and geometric means, we also have another type of lower bound as
	\begin{eqnarray}
	E[\gamma]>\frac{\varepsilon^2}{L_\gamma}\pi^2+\alpha L_\gamma \geq 2\pi\varepsilon\sqrt{\alpha}. \nonumber
	\end{eqnarray}
	Thus the condition in Theorem \ref{thmnooverhang} can be replaced by $4\varepsilon^2\alpha\geq1$.
	Although this condition is meaningful quantitatively, it is not sharp enough to obtain the qualitative property that any minimizer is a graph curve for any fixed $\varepsilon$ and $\alpha\approx1$.
\end{remark}

\section{Overhanging solutions}\label{sectoverhang}

In this section we show that there is a combination of $\varepsilon$, $\alpha$, and $\psi$ such that any minimizer must overhang.

\begin{definition}[Overhanging]
	We say that a curve $\gamma=(x,y)\in\mathcal{A}$ is overhanging if there exists $t\in\bar{I}$ such that $x'(t)<0$.
\end{definition}

\begin{remark}
	By the periodicity of $\gamma\in\mathcal{A}$, there is $t\in\bar{I}$ such that $x'(t)>0$ in general, and thus any overhanging curve must have ``turns'' in the $x$-direction.
\end{remark}

Heuristically, overhanging solutions should appear in order to circumvent sharp mountain folds of substrates as in Figure \ref{figmountainfold} since minimizing curves should bend as the scale $\varepsilon$ in principle.

However, this is a kind of local necessary condition and in general the global shape formation of curves is very complicated.
In order to find overhanging minimizers rigorously, we deal with a special substrate (fakir carpet), which is simple enough to analyze.

\begin{figure}[tb]
	\begin{center}
		\includegraphics[width=50mm]{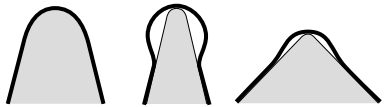}
		\caption{Curves near mountain folds.
			Minimizing curves should bend at the scale $\varepsilon$ and hence overhang to circumvent more sharp folds (center).
			Curves would not overhang for folds with large bending scale (left) or small slope (right).}
		\label{figmountainfold}
	\end{center}
\end{figure}

In what follows, we first give a formal discussion for a very singular substrate, and then rigorously prove the existence of overhanging minimizers for smooth or Lipschitz substrates.

\subsection{For fakir carpets: strategy}

In this subsection we give an intuitive explanation by formally taking a singular substrate as in Figure \ref{figfakircarpet}; $\psi$ is the fakir carpet of height $h$ and period $1$, which is the most simple substrate with a singularly sharp mountain fold (but no longer a continuous function).
For a fakir carpet substrate, we obtain a general lower bound for the energy of all non-overhanging curves and show that, under suitable assumptions on the smallness of $\varepsilon$ and $\alpha$, there is a special overhanging competitor so that its energy is lower than any non-overhanging curve.

\begin{figure}[tb]
	\begin{center}
		\includegraphics[width=60mm]{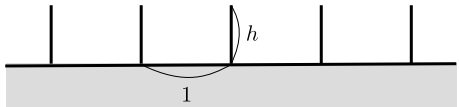}
		\caption{Fakir carpet of height $h$ and the $1$-periodicity.}
		\label{figfakircarpet}
	\end{center}
\end{figure}

We first obtain a lower bound for non-overhanging curves.
In the present setting, it turns out that, for any $\varepsilon$ and $\alpha$, any non-overhanging curve $\gamma\in\mathcal{A}$ satisfies
\begin{equation}\label{curveenergy1}
E[\gamma]\geq\min\{1,h\}.
\end{equation}
In fact, any non-overhanging curve $\gamma$ is either, not touching the base part of the fakir carpet as in Figure \ref{figfakircarpet1}, or touching as in Figure \ref{figfakircarpet2}.
Note that in both cases $\gamma$ touches at most one side of the needle as in the figures since $\gamma$ is not overhanging.
(To touch both sides, the curve must have a singularity.)
In the former case (Figure \ref{figfakircarpet1}), the curve $\gamma$ has the free part of length at least $1$, i.e., $E[\gamma]\geq1$.
In the latter case (Figure \ref{figfakircarpet2}), the curve $\gamma$ has the free part of length at least $h$, i.e., $E[\gamma]\geq h$.
Consequently, any non-overhanging curve satisfies (\ref{curveenergy1}).

\begin{figure}[tb]
	\begin{center}
		\begin{tabular}{cc}
			\begin{minipage}{0.45\hsize}
				\begin{center}
					\includegraphics[width=30mm]{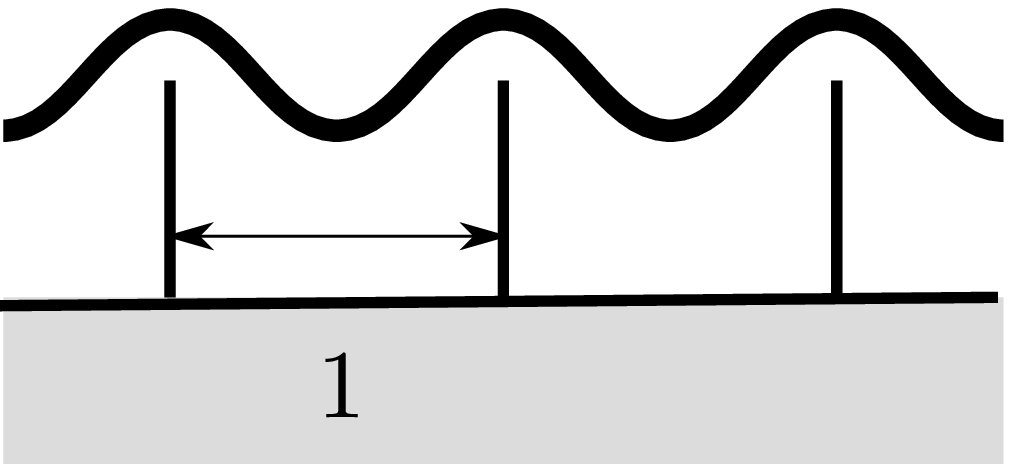}
				\end{center}
			\end{minipage}
			\begin{minipage}{0.45\hsize}
				\begin{center}
					\includegraphics[width=30mm]{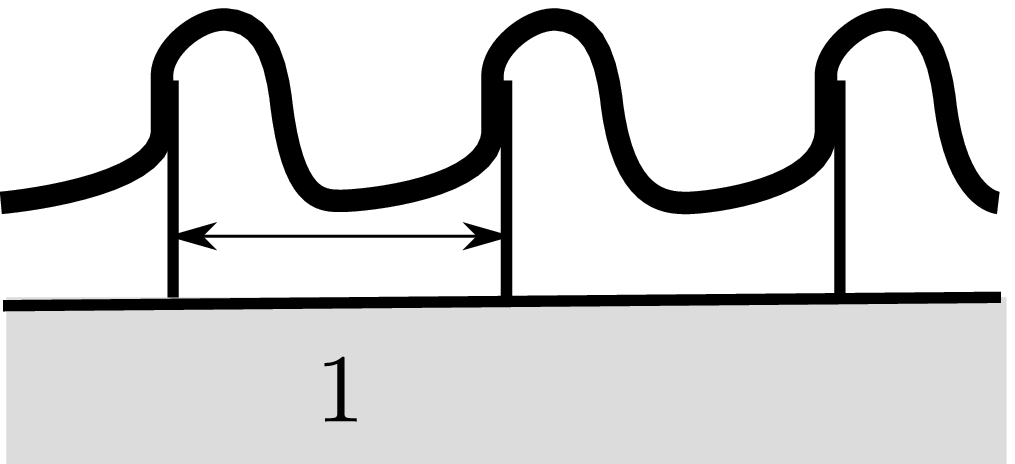}
				\end{center}
			\end{minipage}
		\end{tabular}
		\caption{Non overhanging curves not touching the base part.}
		\label{figfakircarpet1}
	\end{center}
	\begin{center}
		\begin{tabular}{cc}
			\begin{minipage}{0.45\hsize}
				\begin{center}
					\includegraphics[width=30mm]{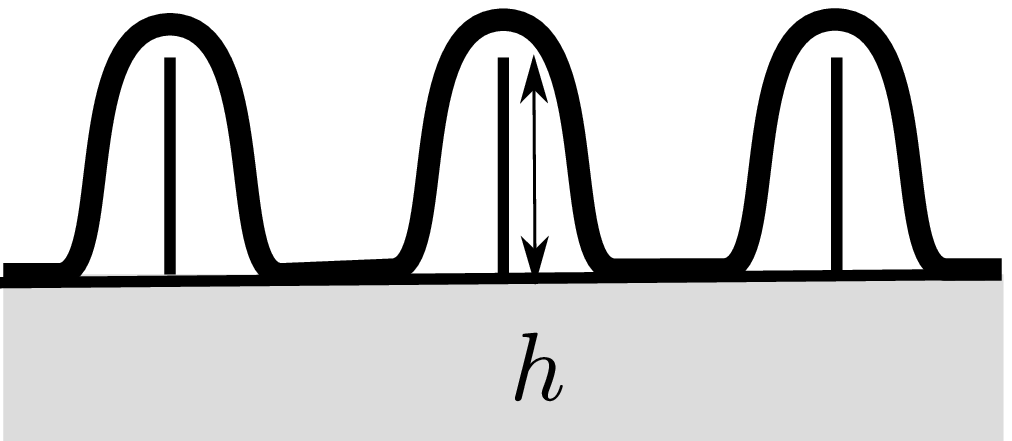}
				\end{center}
			\end{minipage}
			\begin{minipage}{0.45\hsize}
				\begin{center}
					\includegraphics[width=30mm]{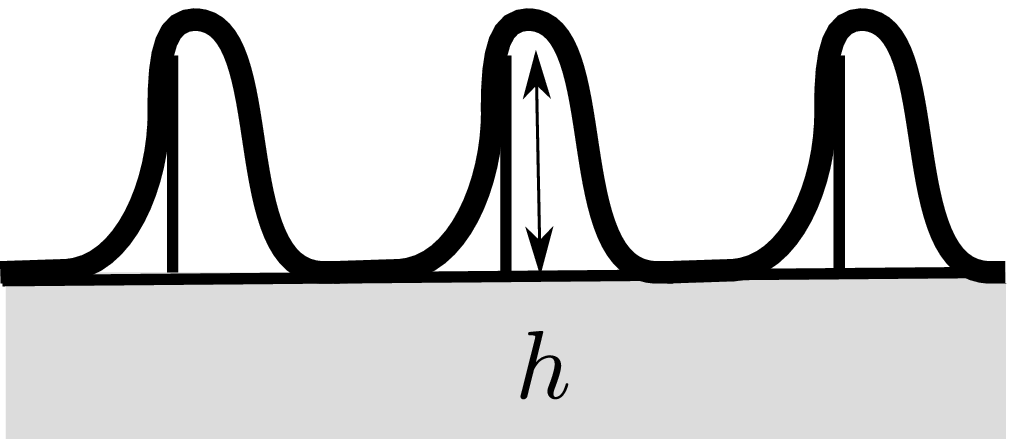}
				\end{center}
			\end{minipage}
		\end{tabular}
		\caption{Non overhanging curves touching the base part.}
		\label{figfakircarpet2}
	\end{center}
\end{figure}

On the other hand, providing that $\varepsilon$ is sufficiently small as $\varepsilon<\min\{1,h\}/5$, we can define an overhanging competitor $\hat{\gamma}\in\mathcal{A}$ as in Figure \ref{figfakircarpetoverhang}, which is almost adhering to the fakir carpet and bending in the free (non-adhering) part as circular arcs of radius $\varepsilon$.
Then $\hat{\gamma}$ satisfies
\begin{equation}\label{curveenergy2}
E[\hat{\gamma}]< (1+2h)\alpha + 20\pi\varepsilon.
\end{equation}
In fact, the total length of the bounded (adhering) part $\hat{\gamma}_B$ is less than $1+2h$, that is, $E[\hat{\gamma}_B]<(1+2h)\alpha$, and in the free part $\hat{\gamma}_F$ the energy $E[\hat{\gamma}_F]$ is bounded as
\begin{eqnarray}
\int_{\hat{\gamma}_F} \left[\varepsilon^2\kappa^2+1\right]ds < 10\pi\varepsilon\left[\varepsilon^2\frac{1}{\varepsilon^2}+1 \right] = 20\pi\varepsilon, \nonumber
\end{eqnarray}
where $10\pi=5\times2\pi$ is a rough upper bound for the total angle of the circular arcs.

\begin{figure}[tb]
	\begin{center}
		\includegraphics[width=50mm]{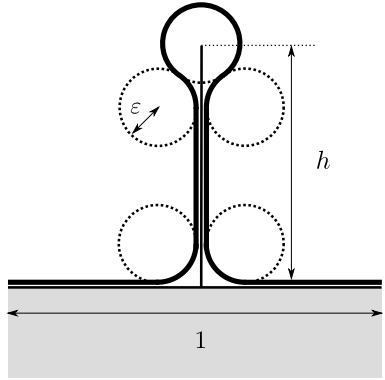}
		\caption{Overhanging competitor above the fakir carpet.
			The curve consists of the adhering straight parts and the non-adhering circular arc parts of radius $\varepsilon$.}
		\label{figfakircarpetoverhang}
	\end{center}
\end{figure}

Combining (\ref{curveenergy1}) and (\ref{curveenergy2}), we see that the conditions
\begin{eqnarray}
\alpha<\Delta:=\frac{\min\{1,h\}}{1+2h}, \quad \varepsilon<\frac{(1+2h)(\Delta-\alpha)}{20\pi} \nonumber
\end{eqnarray}
imply $E[\gamma]>E[\hat{\gamma}]$, which means that the energy of any non-overhanging curve $\gamma$ is strictly higher than the overhanging competitor $\hat{\gamma}$.
In conclusion, for any fakir carpet $\psi$, if $\alpha$ and $\varepsilon$ are sufficiently small as $\alpha\ll\Delta$ and $\varepsilon\ll\min\{1,h\}$, then any minimizer must overhang.

Finally, we remark that for any $h>0$ the inequality
\begin{eqnarray}
\Delta=\frac{\min\{1,h\}}{1+2h}\leq\frac{1}{3} \nonumber
\end{eqnarray}
holds, and the equality is attained if and only if $h=1$.
This means that, at least in our method, the case of height $1$ allows the optimal (highest) upper bound for $\alpha$ or $\varepsilon$ to observe overhangs.
The dimensionless quantity $\Delta$ may be read as a ``deviation'' of the hall of a fakir carpet.
Indeed, the hall is the square when $\Delta$ takes the maximum $1/3$ ($h=1$), and the halls become thin rectangles as $\Delta\downarrow0$ ($h\to0$ or $h\to\infty$).
Thus, the more a hall deviates from the square, the smaller the $\varepsilon$ and $\alpha$ are necessary to be for the presence of overhangs.

\subsection{For smooth substrates}

A similar consideration is valid for a smooth but fakir carpet like substrate $\psi$ as in Figure \ref{figsmoothcarpet}.
The main difference from the singular case is that, in the smooth case, curves may touch both the walls of substrates.
Thus we need to state that if a non-overhanging curve touches both the wall parts of a ``thin'' needle then the total energy is sufficiently high.
To this end, we prepare a general lemma concerning a lower bound for the bending energies of non-overhanging curves as in Figure \ref{figgraphcurve}.
The lower bound only depends on the width of curves in the $x$-direction and the tangential angles at the endpoints.

We define the following nonnegative even function:
\begin{eqnarray}
	f(\theta):=\left|\int_0^{\theta}\sqrt{\cos\varphi} d\varphi\right| = \int_0^{|\theta|}\sqrt{\cos\varphi} d\varphi. \nonumber
\end{eqnarray}
Moreover, for a regular curve $\gamma$, we define the tangential angle $\theta(t)\in(-\pi,\pi]$ at $t\in\bar{I}$ so that $\dot{\gamma}(t)=|\dot{\gamma}(t)|(\cos\theta(t),\sin\theta(t))$.
Then we have the following

\begin{figure}[tb]
	\begin{center}
		\includegraphics[width=40mm]{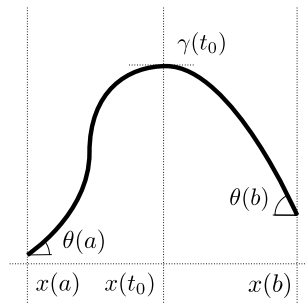}
		\caption{Curve $\gamma=(x,y)$ under the assumption of Lemma \ref{lemma1}.}
		\label{figgraphcurve}
	\end{center}
\end{figure}

\begin{lemma}\label{lemma1}
	Let $J=(a,b)$ be a bounded interval and $\gamma=(x,y)\in H^2(J;\mathbb{R}^2)$ be a regular curve such that $x'(t)\geq0$ for any $t\in J$ and $y'(t_0)=0$ for some $t_0\in J$.
	Then the following inequality holds:
	\begin{eqnarray}
		\int_{\gamma}\kappa^2ds \geq \frac{[f(\theta(a))+f(\theta(b))]^2}{x(b)-x(a)}. \nonumber
	\end{eqnarray}
	%In particular, in the case that $|\theta(a)|=|\theta(b)|=\pi/2$, the lower bound is larger than $\frac{4}{x(b)-x(a)}$.
\end{lemma}

\begin{proof}
	We first assume that a curve $\gamma=(x,y)$ satisfies $x'(t)>0$ for any $t\in\overline{J}$.
	Then the curve is represented by some function $u\in H^2(x(a),x(b))$ such that $u'(x(a))=\tan\theta(a)$, $u'(x(b))=\tan\theta(b)$, $u'(x(t_0))=0$ and
	\begin{eqnarray}
		\int_{\gamma}\kappa^2ds=\int_{x(a)}^{x(b)}\frac{|u''(z)|^2}{(1+u'(z)^2)^{5/2}}dz. \nonumber
	\end{eqnarray}
	The Cauchy-Schwarz inequality yields that
	\begin{eqnarray}
		&&\int_{x(a)}^{x(b)}\frac{|u''(z)|^2}{(1+u'(z)^2)^{5/2}}dz \nonumber\\
		&\geq& \frac{1}{x(b)-x(a)}\left(\int_{x(a)}^{x(b)}\frac{|u''(z)|}{(1+u'(z)^2)^{5/4}}dz\right)^2. \nonumber
	\end{eqnarray}
	By change of variables, we have
	\begin{eqnarray}
		&&\int_{x(a)}^{x(b)}\frac{|u''(z)|}{(1+u'(z)^2)^{5/4}}dz \nonumber\\
		&\geq& \left|\int_{\tan\theta(a)}^{0}\frac{dw}{(1+w^2)^{5/4}}\right| + \left|\int_{0}^{\tan\theta(b)}\frac{dw}{(1+w^2)^{5/4}}\right| \nonumber\\
		&=& f(\theta(a))+f(\theta(b)), \nonumber
	\end{eqnarray}
	and thus we obtain the desired lower bound.
	
	For general $\gamma=(x,y)$ with $x'\geq0$, we obtain the same conclusion by considering modified curves as $\gamma_\delta(t)=(x(t)+\delta t,y(t))$ for small $\delta>0$ and taking the limit $\delta\to0$.
	Note that $x(b)>x(a)$ holds even in this case by the assumption of $\gamma$.
\end{proof}

We now state and prove the main theorem.
Let $h>0$ and $0<2\delta<\min\{1,h\}$.
A $1$-periodic function $\phi$ is called $\delta$-smooth fakir carpet of height $h$ if $\phi$ is as in Figure \ref{figsmoothcarpet}, namely, of class $C^\infty$ and satisfies
\begin{enumerate}
	\item $\phi(x)=\phi(1-x)$ for any $x\in[0,1/2]$,
	\item $\phi\equiv0$ in $[0,1/2-\delta]$ and $\phi(1/2)=h$,
	\item $\phi'\geq0$ in $[0,1/2]$, 
	\item $\phi''(x)=0$ while $\delta\leq\phi(x)\leq h-\delta$.
\end{enumerate}
Moreover, we define its base and wall parts as in Figure \ref{figsmoothparts}; the base part is the part with $y=\phi(x)=0$ and the left (resp.\ right) wall part is the part with $y=\phi(x)$, $\delta\leq y\leq h-\delta$ and $\phi'(x)>0$ (resp.\ $\phi'(x)<0$).
All the parts are straight.
Note that $\delta\gtrsim \|\phi''\|_\infty^{-1}$.

\begin{theorem}\label{thmoverhanging}
	Let $h>0$ and $\Delta:=\frac{\min\{1,h\}}{1+2h}$.
	Then for any $\alpha<\Delta$ and $\varepsilon<\frac{(1+2h)(\Delta-\alpha)}{20\pi}$ there exists $0<\bar{\delta}<\varepsilon$ such that, for any $\delta$-smooth fakir carpet substrate $\psi_\delta$ of height $h$ with $0<\delta<\bar{\delta}$, any minimizer of (\ref{minimizingproblem}) is overhanging.
\end{theorem}

\begin{figure}[tb]
	\begin{center}
		\includegraphics[width=60mm]{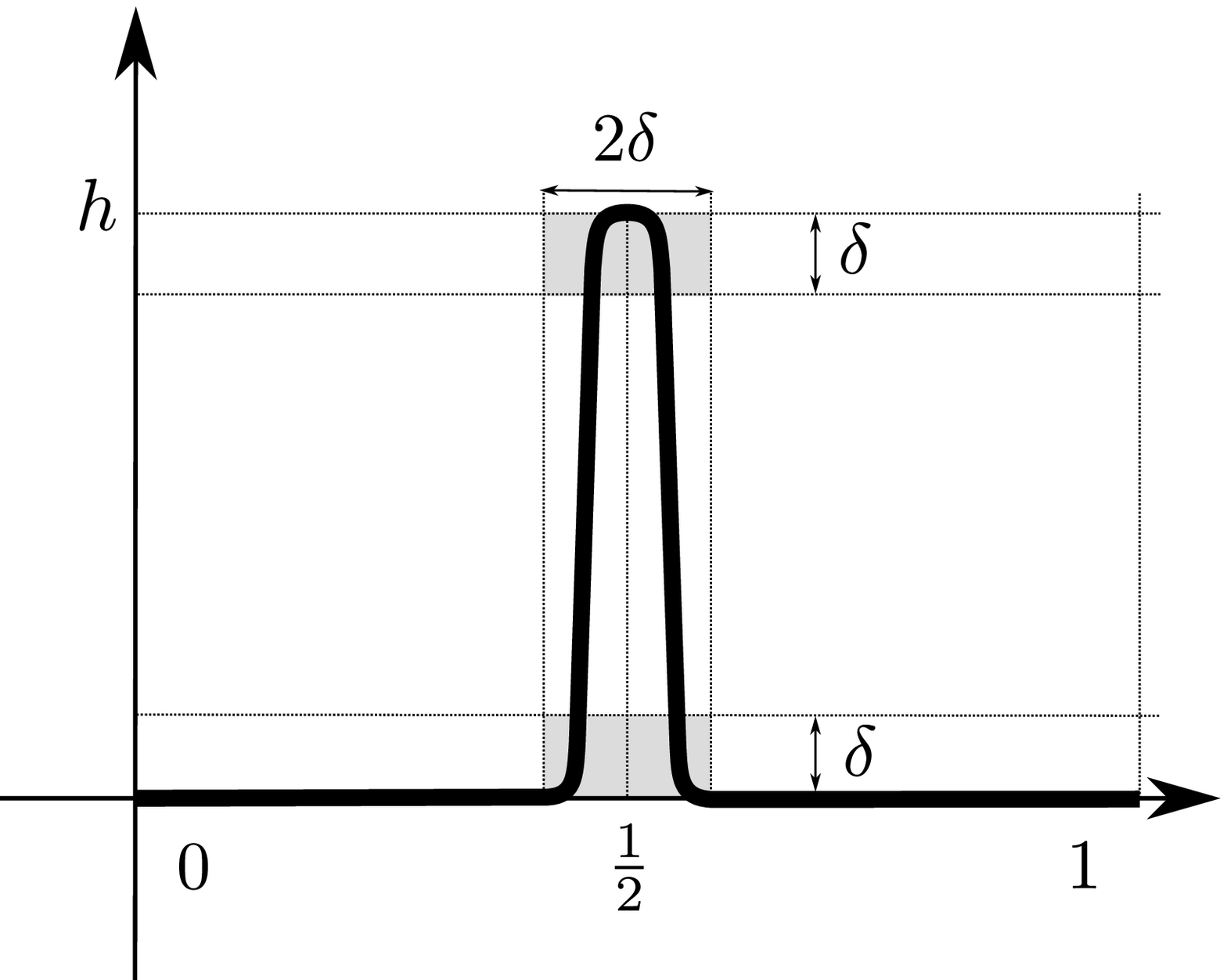}
		\caption{
			A $1$-period part of the smooth but fakir carpet like substrate.
			The function is bending only in the gray regions and otherwise straight.}
		\label{figsmoothcarpet}
	\end{center}
	\begin{center}
		\includegraphics[width=60mm]{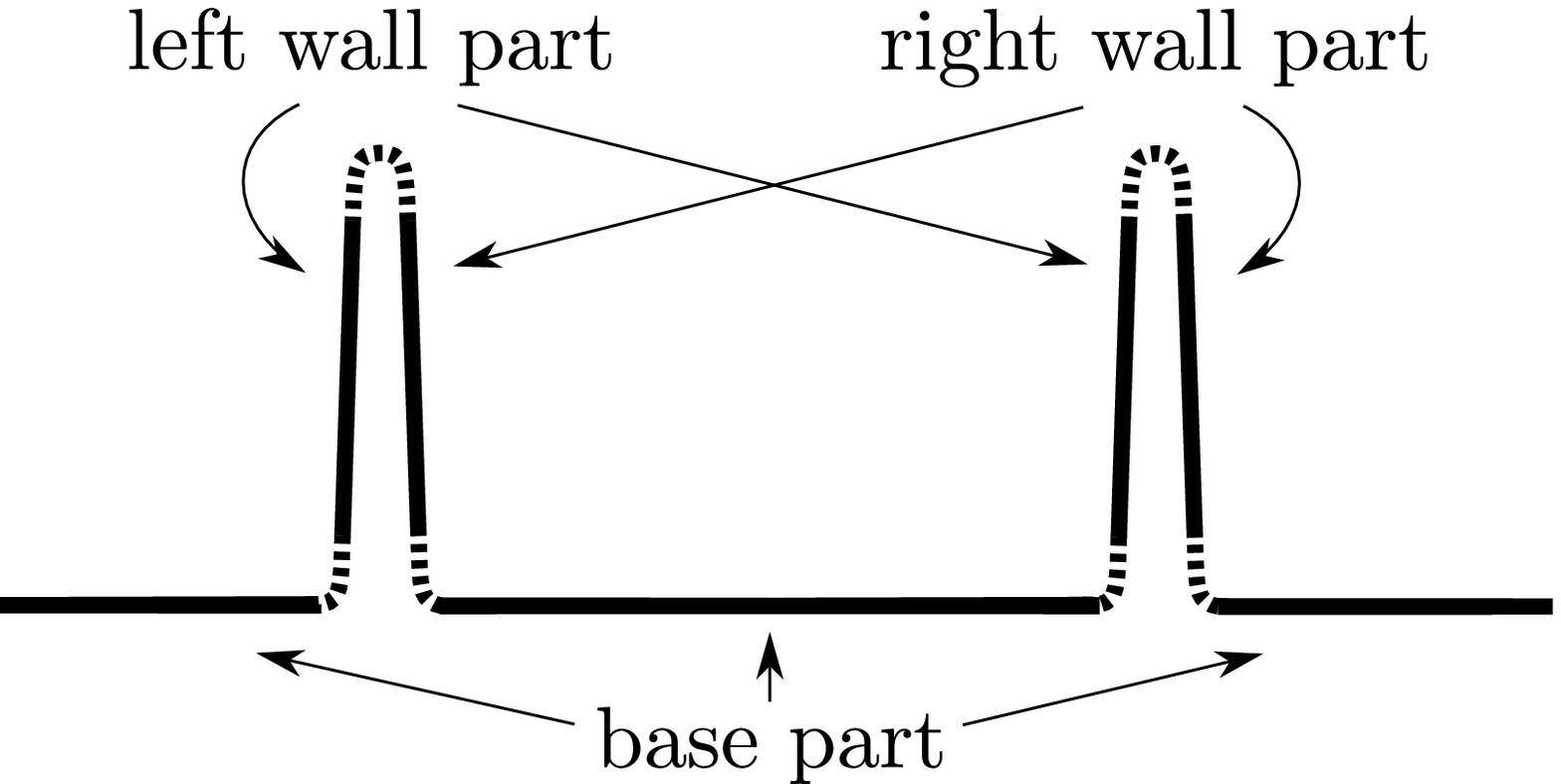}
		\caption{Base part and left and right wall parts of $\delta$-smooth fakir carpets.}
		\label{figsmoothparts}
	\end{center}
\end{figure}

\begin{figure}[tb]
	\begin{center}
		\includegraphics[width=50mm]{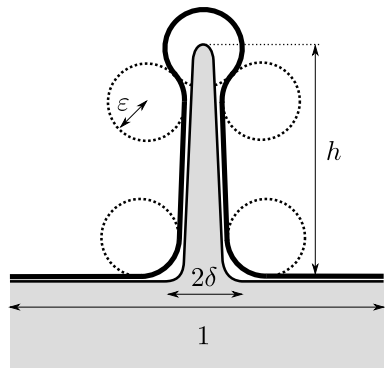}
		\caption{
			Overhanging competitor $\hat{\gamma}_\delta$ above the smooth substrate $\psi_\delta$ defined for any small $\delta\ll\varepsilon$.
			The curve $\hat{\gamma}_\delta$ is adhering to the substrate only in the base and wall parts and otherwise bending as circular arc of radius $\varepsilon$.
			For any small $\delta$ the curve $\hat{\gamma}_\delta$ is overhanging.}
		\label{figsmoothoverhang}
	\end{center}
\end{figure}

\begin{figure}[tb]
	\begin{center}
		\includegraphics[width=60mm]{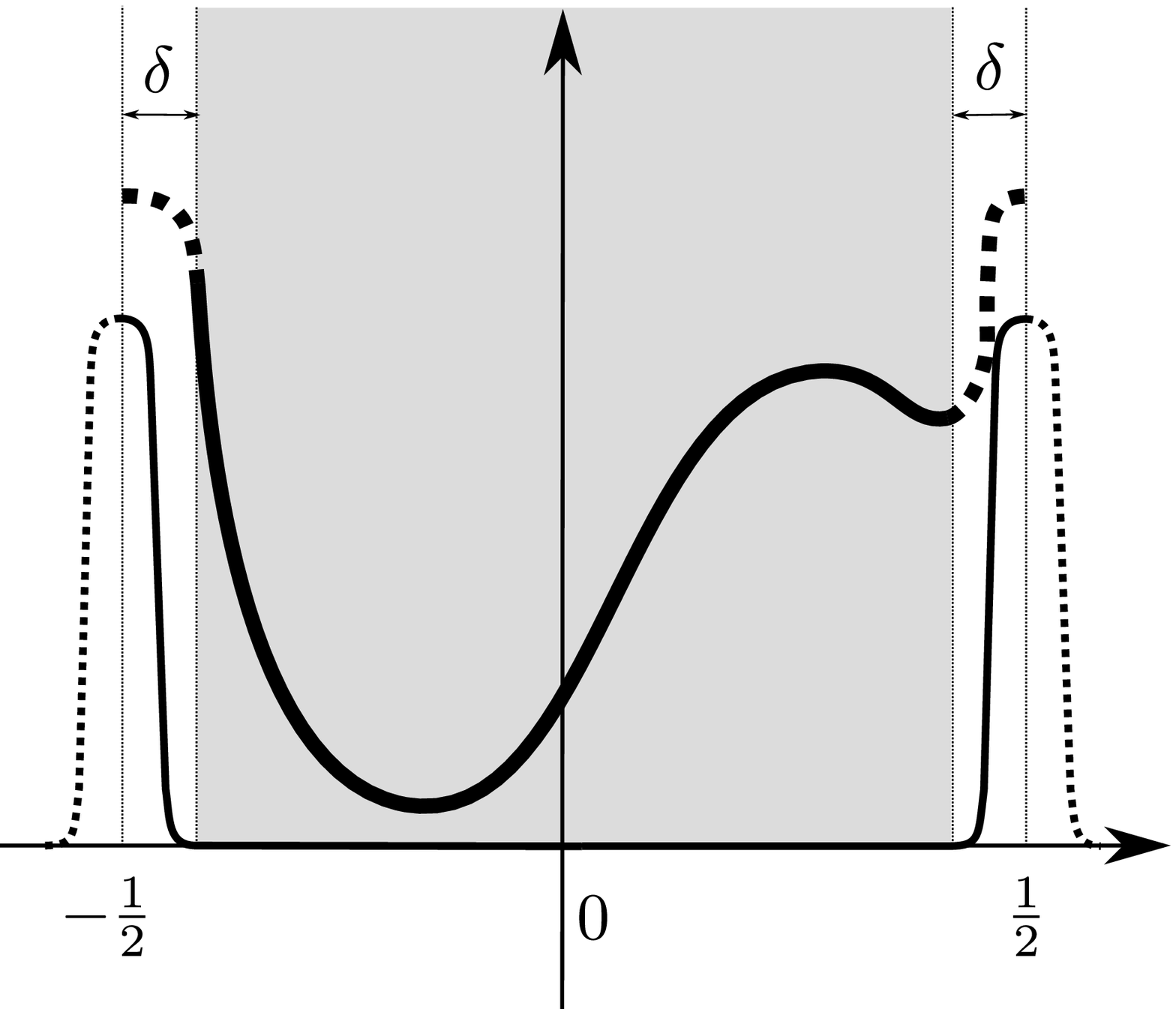}
		\caption{A curve not touching the base part in $(-\frac{1}{2},\frac{1}{2})$.
			To avoid the base part, any periodic curve $\gamma\in\mathcal{A}$ must cross the gray region without touching the graph of $\psi_\delta$.}
		\label{figsmoothproof1}
	\end{center}
	\begin{center}
		\includegraphics[width=60mm]{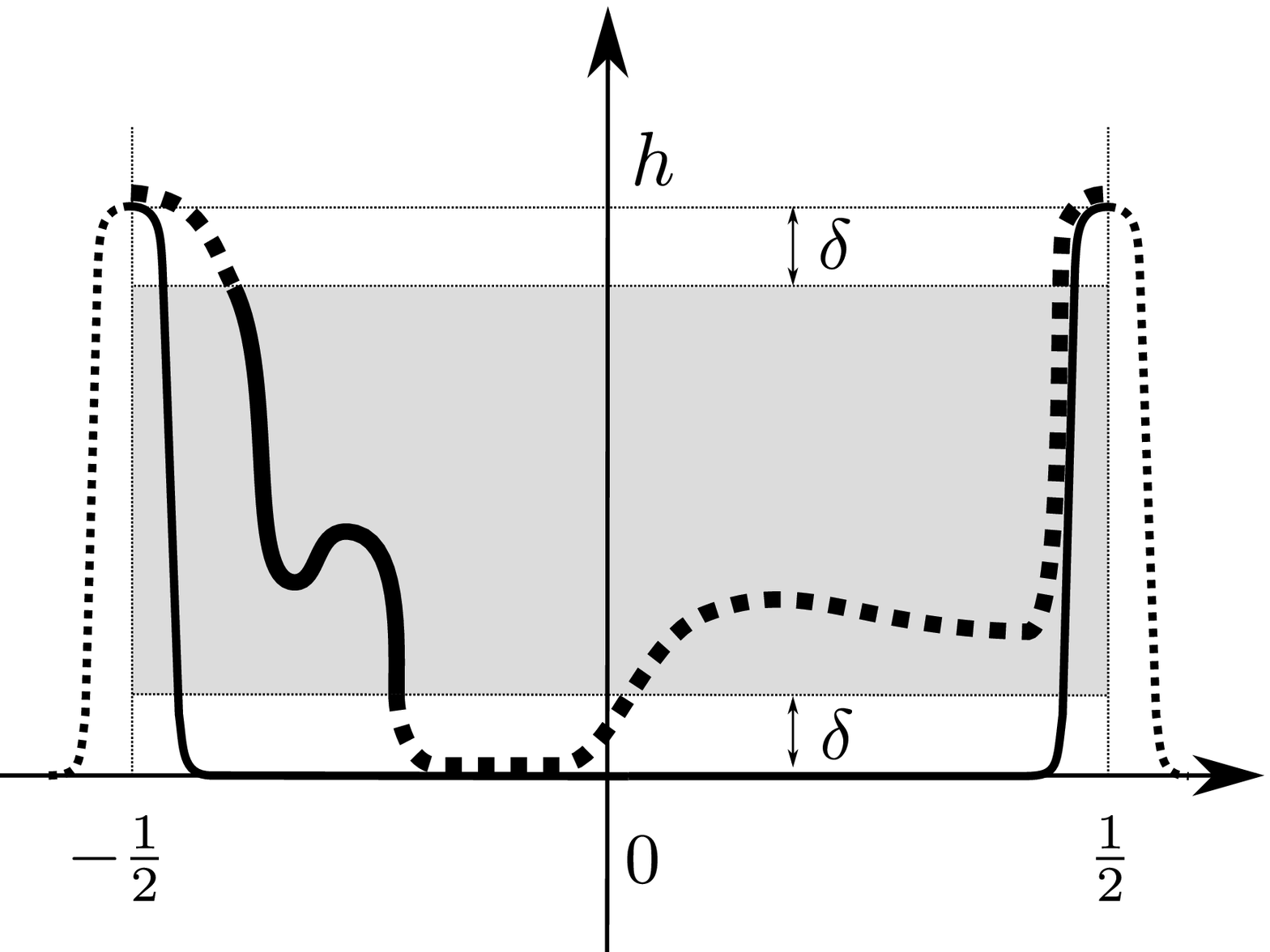}
		\caption{A curve touching the base part but avoiding the right wall part in $(-\frac{1}{2},\frac{1}{2})$.
			To touch the base part and avoid the right (or left) wall part, any periodic non-overhanging curve $\gamma\in\mathcal{A}$ must cross the gray region without touching the graph of $\psi_\delta$ at least one time.}
		\label{figsmoothproof2}
	\end{center}
\end{figure}

\begin{proof}
	Fix any $\alpha<\Delta$ and $\varepsilon<\frac{(1+2h)(\Delta-\alpha)}{20\pi}$.
	For any small $0<\delta\ll\varepsilon$, we take a substrate $\psi_\delta$ of $\delta$-smooth fakir carpet of height $h$.
	Take the overhanging competitor $\hat{\gamma}_\delta\in\mathcal{A}$ as in Figure \ref{figsmoothoverhang}.
	Then, by the similar way to obtain (\ref{curveenergy2}), we see that
	\begin{eqnarray}
		E[\hat{\gamma}_\delta]\leq (1+2h)\alpha+20\pi\varepsilon.\nonumber
	\end{eqnarray}
	By the assumptions on $\alpha$ and $\varepsilon$, we have
	\begin{eqnarray}
		E[\hat{\gamma}_\delta]\leq\min\{1,h\}-c, \nonumber
	\end{eqnarray}
	where $c>0$ is some constant independent of $\delta$.
	Therefore it suffices to prove that
	\begin{eqnarray}\label{infnonoverhang}
		\lim_{\delta\downarrow0}\inf_\gamma E[\gamma]\geq\min\{1,h\},
	\end{eqnarray}
	where the infimum is taken over all non-overhanging curves in the upper side of $\psi_\delta$.
	Indeed, if this is proved then there exists $0<\bar{\delta}<\varepsilon$ such that for any $0<\delta<\bar{\delta}$ any non-overhanging curve $\gamma$ satisfies $E[\hat{\gamma}_\delta]<E[\gamma]$.
	
	Notice that any (non-overhanging) curve $\gamma\in\mathcal{A}$ belongs to at least one of the following three cases:
	\begin{enumerate}
		\item $\gamma$ does not touch the base part (Figure \ref{figsmoothproof1}),
		\item $\gamma$ touches the base part but not the left nor right wall part (Figure \ref{figsmoothproof2}),
		\item $\gamma$ touches both the left and right wall parts.
	\end{enumerate}
	We prove (\ref{infnonoverhang}) for all the cases 1, 2, and 3.
	
	{\it Case 1}.
	By the periodicity, as in Figure \ref{figsmoothproof1}, any curve $\gamma=(x,y)\in\mathcal{A}$ may be regarded as satisfying $x(0)=-1/2$ and $x(1)=1/2$.
	Then the condition of Case 1 implies that $\gamma$ passes through the region $\{|x|<1/2-\delta\}$ freely.
	Hence we have $E[\gamma] \geq 1-2\delta$, which implies (\ref{infnonoverhang}).
	
	{\it Case 2}.
	Similarly, as in Figure \ref{figsmoothproof2}, we may regard any curve $\gamma=(x,y)\in\mathcal{A}$ as satisfying $x(0)=-1/2$ and $x(1)=1/2$, and hence $y(0)=y(1)\geq h$.
	Then the condition of Case 2 and the fact that $\gamma$ is non-overhanging imply that $\gamma$ passes through the region $\{\delta<y<h-\delta\}$ freely at least one time.
	Hence we have $E[\gamma]\geq h-2\delta$,
	which implies (\ref{infnonoverhang}).
	
	{\it Case 3}.
	For any non-overhanging $\gamma\in\mathcal{A}$ touching both the wall parts (tangentially), there are $t_1,t_2\in I$ such that the part of $\gamma$ from $t_1$ to $t_2$ satisfies the assumption of Lemma \ref{lemma1} with $x(t_2)-x(t_1)\leq2\delta$ and $|\theta(t_1)|=|\theta(t_2)|=\theta_\delta$, where $\theta_\delta>0$ is the slope angle of the left wall part.
	Then Lemma \ref{lemma1} implies that
	\begin{eqnarray}
		E[\gamma]\geq\varepsilon^2\int_{\gamma}\kappa^2ds\geq \varepsilon^2\frac{4f(\theta_\delta)^2}{x(t_2)-x(t_1)}\geq \frac{2\varepsilon^2 f(\theta_\delta)^2}{\delta} \nonumber
	\end{eqnarray}
	and especially (\ref{infnonoverhang}).
	The proof is now complete.
\end{proof}
 
Theorem \ref{thmoverhanging} indicates that the smallness of the height of $\psi$ does not imply the graph representations of minimizers.
In this view, we can simplify the statement as

\begin{corollary}
	For any $h>0$, there exist $\varepsilon$, $\alpha$, and smooth $\psi$ of height $h$ such that any minimizer of (\ref{minimizingproblem}) must overhang.
\end{corollary}

In addition, as mentioned in the previous subsection, $h=1$ gives the optimal upper bound $1/3$ for $\alpha$ in our method.
In this view, Theorem \ref{thmoverhanging} is simplified as

\begin{corollary}
	For any $0<\alpha<1/3$, there exist $\varepsilon$ and smooth $\psi$ such that any minimizer of (\ref{minimizingproblem}) must overhang.
\end{corollary}

\subsection{For Lipschitz substrates}

Finally, for small $\alpha$, we give an example of a Lipschitz (singularly folding) substrate with large slope such that any minimizer must be overhanging for ``any'' small $\varepsilon$.
This kind of uniformity is mathematically important.
An intuitive meaning of this result has been given in Introduction.

We shall state it as a proposition.
Let $h>0$ and $0<2\delta<\min\{1,h\}$.
A $1$-periodic function $\phi$ is called $\delta$-Lipschitz fakir carpet of height $h$ if
\begin{eqnarray}
\phi(x):=\max\left\{0,h-\left|\frac{h}{\delta}x-\frac{1}{2}\right|\right\} \nonumber
\end{eqnarray}
for $x\in[0,1]$.
We also define the base and wall parts as well as the smooth case; namely, the base part is the part with $y=\phi(x)=0$ and the left (resp.\ right) part is the part with $y=\phi(x)$, $\delta<y<h-\delta$ and $\phi(x)'>0$ (resp.\ $\phi'(x)<0$).

\begin{theorem}\label{thmoverhanging2}
	Let $h>0$ and $\alpha<\Delta:=\frac{\min\{1,h\}}{1+2h}$.
	Then there exist $\bar{\varepsilon}>0$ and $\bar{\delta}>0$ such that, for any $0<\varepsilon<\bar{\varepsilon}$ and the $\delta$-Lipschitz fakir carpet substrate $\psi_\delta$ of height $h$ with any $0<\delta<\bar{\delta}$, any minimizer of (\ref{minimizingproblem}) is overhanging.
\end{theorem}

\begin{figure}[tb]
	\begin{center}
		\includegraphics[width=40mm]{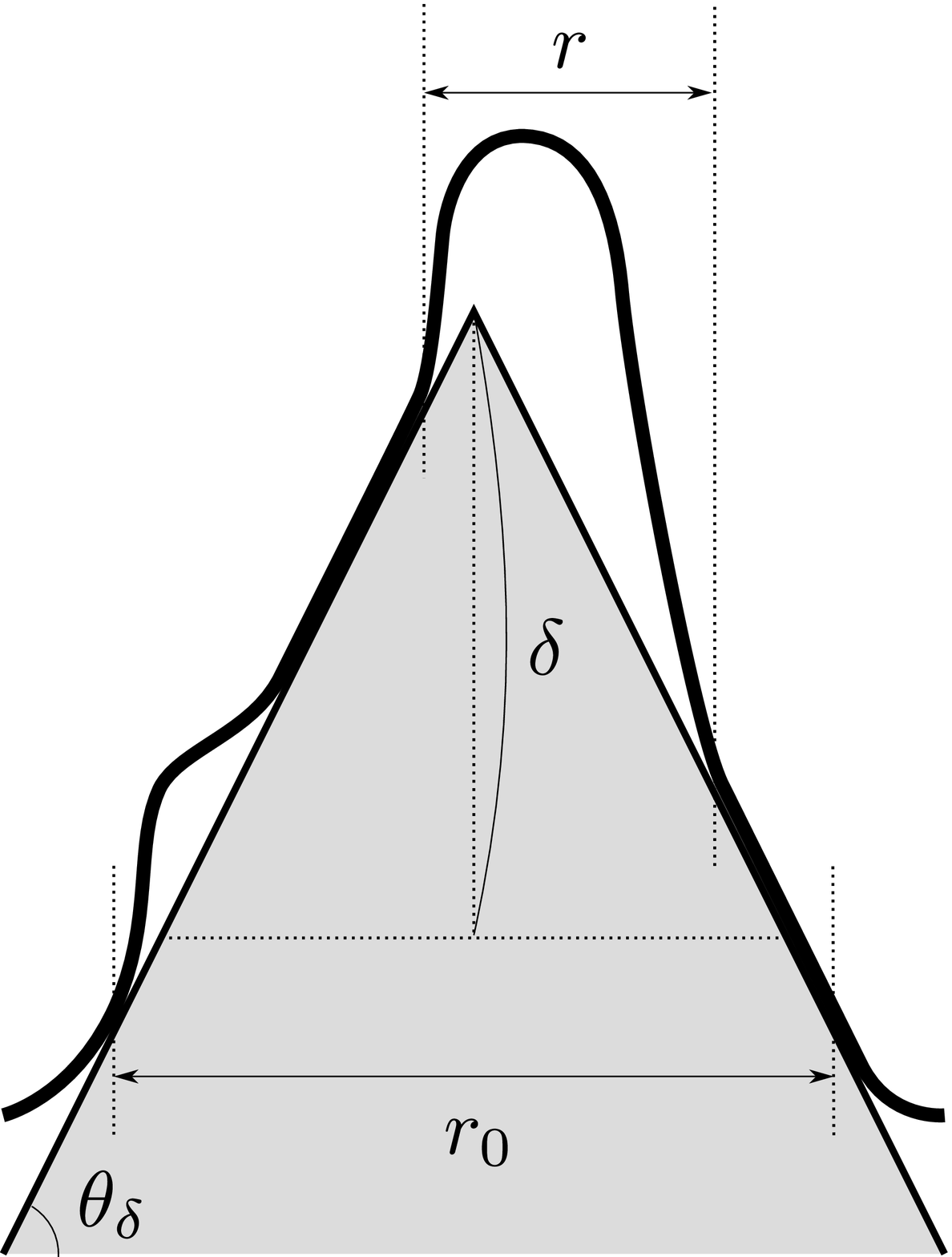}
		\caption{Non overhanging curve $\gamma$ touching the left and right wall parts.
			There are two touching points $\gamma(t_1)$ and $\gamma(t_2)$ of height less than $h-\delta$.
			There are also two points such that $\gamma$ is tangent to the left and right slope there but does not touch $\psi_\delta$ between them.}
		\label{figLipchitzproof1}
	\end{center}
	\begin{center}
		\includegraphics[width=40mm]{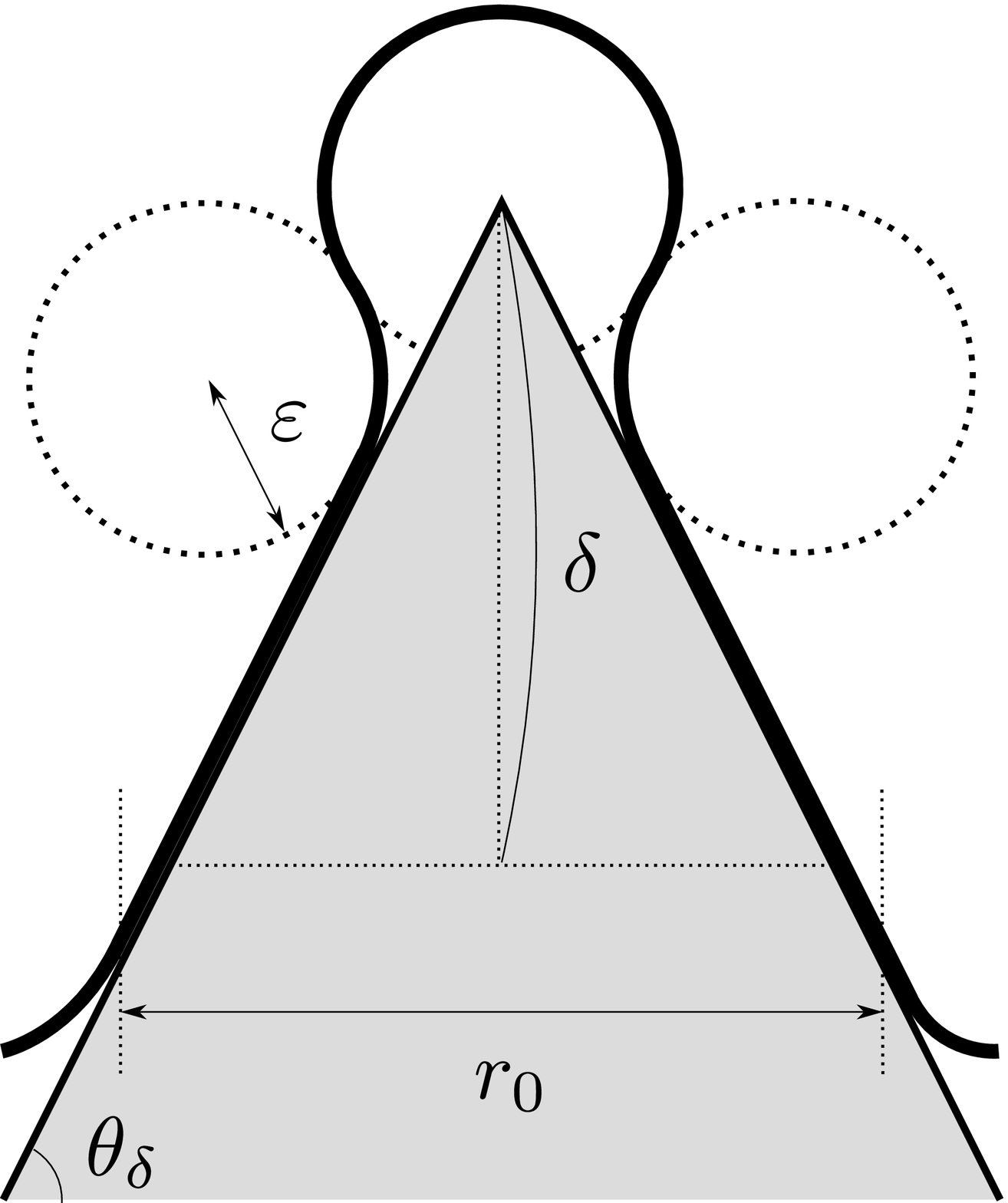}
		\caption{The curve of Figure \ref{figLipchitzproof1} modified from $\gamma(t_1)$ to $\gamma(t_2)$.
			The modified part consists of straightly adhering parts and a freely bending part with radius $\varepsilon$.
			This curve is well-defined and not self-intersecting whenever $\varepsilon\ll\delta$ and moreover overhanging whenever $\theta_\delta>\pi/3$.}
		\label{figLipchitzproof2}
	\end{center}
\end{figure}

\begin{proof}
	Noting the condition of $\alpha$, in the same way as Case 1 and Case 2 in Theorem \ref{thmoverhanging}, we see that there are $\delta_0>0$ and $\varepsilon_0>0$ such that, for any $0<\delta<\delta_0$ and $0<\varepsilon<\varepsilon_0$, any non-overhanging curve is necessary to touch the left and right wall parts in order to minimize $E$.
	Note that for arbitrary small $\varepsilon>0$ an overhanging competitor as in Figure \ref{figsmoothoverhang} is well-defined since the substrate $\psi_\delta$ is folding singularly.
	
	To complete the proof, we shall prove that, for any small $\delta$, $\varepsilon$, and any non-overhanging $\gamma$ touching both the wall parts, there is an overhanging competitor $\hat{\gamma}$ such that $E[\gamma]>E[\hat{\gamma}]$.
	
	Fix arbitrary $0<\delta<\delta_0$ and $0<\varepsilon<\varepsilon_0$ and take any non-overhanging $\gamma=(x,y)\in\mathcal{A}$ touching both the wall parts.
	Then there are times $t_1<t_2$ such that $\gamma$ touches the left (resp.\ right) wall part at $t_1$ (resp.\ $t_2$).
	Define $t_3\in[t_1,t_2]$ (resp.\ $t_4\in[t_1,t_2]$) as the supremum (resp.\ infimum) of time $t\in[t_1,t_2]$ such that $\gamma(t)$ touches $\psi_\delta$ and $x(t),1/2$ (resp.\ $x(t)>1/2$).
	Note that $\theta(t_3)=-\theta(t_4)=\theta_\delta$, where $\theta_\delta>0$ denotes the slope angle of $\psi_\delta$.
	Moreover, in $(t_3,t_4)$ the curve $\gamma$ does not touch $\psi_\delta$ except at the vertex $(1/2,\psi_\delta(1/2))$.
	Denote $r_0=x(t_2)-x(t_1)$ and $r=x(t_4)-x(t_3)$ as in Figure \ref{figLipchitzproof1}.
	
	Then, by Lemma \ref{lemma1} and the fact that $\gamma$ circumvents the vertex of $\psi_\delta$ freely (except the vertex), the energy $E$ of the part of $\gamma$ from $t_3$ to $t_4$ is bounded below as
	\begin{eqnarray}
		E[\gamma|_{[t_3,t_4]}]\geq\varepsilon^2\frac{4f(\theta_\delta)^2}{r}+\frac{r}{\cos\theta_\delta}. \nonumber
	\end{eqnarray}
	In addition, the part from $t_1$ to $t_3$ and from $t_4$ to $t_2$ is totally bounded below as
	\begin{eqnarray}
		E[\gamma|_{[t_1,t_3]}]+E[\gamma|_{[t_4,t_2]}]\geq\alpha\cdot\frac{r_0-r}{\cos\theta_\delta} \nonumber
	\end{eqnarray}
	since the energies of $\gamma|_{[t_1,t_3]}$ and $\gamma|_{[t_4,t_2]}$ are more than or equal to the energies of the completely adhering straight lines joining the endpoints of $\gamma|_{[t_1,t_3]}$ and $\gamma|_{[t_4,t_2]}$, respectively.
	Therefore, the part of $\gamma$ from $t_1$ to $t_2$ is bounded below as
	\begin{eqnarray}
		E[\gamma|_{[t_1,t_2]}]&\geq& \frac{\varepsilon^24f(\theta_\delta)^2}{r}+\frac{r}{\cos\theta_\delta}+\frac{\alpha(r_0-r)}{\cos\theta_\delta} \nonumber\\ 
		&=&  \varepsilon^24f(\theta_\delta)^2\frac{1}{r}+\frac{1-\alpha}{\cos\theta_\delta}r+\frac{\alpha r_0}{\cos\theta_\delta} \nonumber\\
		&\geq&  \left(2\sqrt{4f(\theta_\delta)^2\frac{1-\alpha}{\cos\theta_\delta}}\right)\varepsilon+\frac{\alpha r_0}{\cos\theta_\delta}, \nonumber
	\end{eqnarray}
	which does not depend on $r$.
	
	On the other hand, providing that $\delta$ and $\varepsilon$ are sufficiently small as $\theta_\delta>\pi/3$ and $\varepsilon<\delta/3$, the competitor $\hat{\gamma}$ constructed by modifying $\gamma$ in $(t_1,t_2)$ as in Figure \ref{figLipchitzproof2} is well-defined and overhanging.
	The energy of $\hat{\gamma}$ from $t_1$ to $t_2$ is bounded above as
	\begin{eqnarray}
		E[\hat{\gamma}|_{[t_1,t_2]}] &\leq& 6\pi\varepsilon\left[\varepsilon^2\frac{1}{\varepsilon^2}+1 \right] +\alpha\cdot\frac{r_0}{\cos\theta_\delta} \nonumber\\
		&=& 12\pi\varepsilon + \frac{\alpha r_0}{\cos\theta_\delta}. \nonumber
	\end{eqnarray}
	In the outside of $(t_1,t_2)$, the curves $\hat{\gamma}$ and $\gamma$ coincide.
	
	Consequently, noting that for any small $\delta>0$
	\begin{eqnarray}
	2\sqrt{4f(\theta_\delta)^2\frac{1-\alpha}{\cos\theta_\delta}} > 12\pi, \nonumber
	\end{eqnarray}
	we have $E[\gamma]>E[\hat{\gamma}]$ for any small $\delta$ and $\varepsilon$.
	The proof is now complete.
\end{proof}

\section{Discussion}\label{sectdiscuss}

In this section, we give some further remarks and discussions.

\subsection{Small bending scale}

We first discuss the graph representations of minimizers for small bending scale $\varepsilon$.
Recall that Theorem \ref{thmnooverhang} states that large bending scale $\varepsilon\gg1$ implies the graph representation independently of $\psi$.
This theorem is relatively easy to prove since, if $\varepsilon\gg1$, the periodic boundary condition is effective and hence non-graph curves must have large energies.
On the other hand, the case that $\varepsilon\ll1$ is not easy to obtain the graph representation rigorously since there is no large difference in the energies of graph and non-graph curves.
In fact, Theorem \ref{thmoverhanging} and Theorem \ref{thmoverhanging2} state that overhanging minimizers exist when $\varepsilon$ is small and the minimal bending scale $r=\|\psi''\|_\infty^{-1}$ of $\psi$ is much smaller.
However, we may expect the graph representation when $\varepsilon\ll r$ by the following formal observation (cf.\ \cite{Mi16}).
When $\psi$ is smooth and $\varepsilon=0$, minimizers would be Lipschitz functions with straight free parts as in Figure \ref{figvalleyfold}.
They have edge singularities at the contact points as valley folds and the minimal distance $d$ among the singularities is bounded below (by a constant depending on $\alpha$ and $\psi$).
Thus, for small $\varepsilon\ll \min\{d,r\}$, any minimizer would be obtained by modifying such valley folds smoothly as in Figure \ref{figvalleyfold}.
Moreover, in contrast to mountain folds (Figure \ref{figmountainfold}), the modification of valley folds would not require to increase the slopes.
Hence any minimizer is expected to be a graph curve.

\begin{figure}[tb]
	\begin{center}
		\includegraphics[width=60mm]{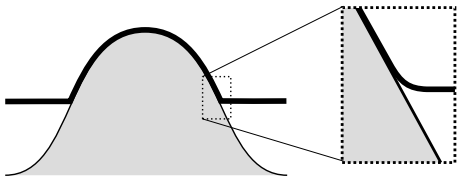}
		\caption{Minimizer on a smooth substrate for small $\varepsilon$.
			If $\varepsilon=0$ then the minimizer has valley fold singularities as the left.
			When $\varepsilon$ is small the singularities would be modified as smooth curves with scale $\varepsilon$ as the right.}
		\label{figvalleyfold}
	\end{center}
\end{figure}

\subsection{Flat substrates}

Theorem \ref{thmnooverhang2} states that a second order flatness of $\psi$ implies the graph representations of minimizers; for any $\varepsilon>0$ and $0<\alpha<1$ there is $k>0$ such that if $\psi$ satisfies $\|\psi''\|_\infty\leq k$ then any minimizer is a graph curve.
The problem would become more difficult if we replace $\psi''$ with $\psi'$ or $\psi$.

Another interesting problem is the following uniform and strong version:
is there $k>0$ such that for any $\varepsilon$, $\alpha$, and $\psi$ with $\|\psi''\|_\infty\leq k$ or $\|\psi'\|_\infty\leq k$ any minimizer is a graph curve?
Notice that Theorem \ref{thmoverhanging2} states that any smallness of $\|\psi\|_\infty$ does not imply the above conclusion.

\subsection{On self-intersections}

Self-intersections are more difficult to occur than overhangs in the sense that any self-intersecting curve in $\mathcal{A}$ must be overhanging.
Here we say that $\gamma\in\mathcal{A}$ has a self-intersection if there are $0\leq t_1<t_2<1$ and $m\in\mathbb{Z}$ such that $\gamma(t_1)=\gamma(t_2)+(m,0)\in\mathbb{R}^2$.
This definition is suitable for our periodic setting; if we take $\tilde{\gamma}\in H^2_{loc}(\mathbb{R};\mathbb{R}^2)$ such that $\tilde{\gamma}(t+1)=\tilde{\gamma}(t)+(1,0)$ for any $t\in\mathbb{R}$, then $\tilde{\gamma}$ is not injective if and only if the restriction $\gamma=\tilde{\gamma}|_I\in\mathcal{A}$ has a self-intersection in the above sense.
This paper proves the existence of overhanging minimizers, but it is not clear whether there exist self-intersecting minimizers in our setting.

Our admissible curves $\gamma\in\mathcal{A}$ may have any self-intersection (self-contact and self-crossing) as mentioned, and thus our problem would be a priori suitable only for filaments but not membranes.
To make our problem compatible with membranes, we especially need to exclude self-crossing curves.
To this end, let $\mathcal{A}_*\subset\mathcal{A}$ be the $H^2$-weak closure of the set of curves without self-intersection.
This set is compatible with membranes since $\mathcal{A}_*$ consists of non self-intersecting curves and limits of such curves at least in $C^1$; especially, any curve of $\mathcal{A}_*$ is not self-crossing but only self-contacting.
Then all the results in this paper are valid even if we replace $\mathcal{A}$ with $\mathcal{A}_*$ in the problem (\ref{minimizingproblem}) since all the competitors used in our proof have no self-intersection.
The existence of minimizers is proved in the same way as \ref{appendix1} since $\mathcal{A}_*$ is $H^2$-weakly closed.

We mention that this kind of self-contact setting has been considered in e.g.\ \cite{DaMaNo15,DoMuRo11}.
In particular, the paper \cite{DaMaNo15} proves that, for confined closed free elasticae, (i) any convex confinement admits only convex minimizers, which especially have no self-intersection, and (ii) there is a confinement with two halls which admits a self-contacting minimizer.
These results indicate that whether minimizers have self-contacts crucially relates to the simplicity of confinements.
Our graph confinements are simple but the effect of adhesion make curves easier to form complicated shapes, and hence the self-intersection problem becomes more nontrivial.

\subsection{Local minimizers}

We next give a theoretical discussion on local minimizers.
A curve $\gamma_0\in\mathcal{A}$ is called a local minimizer if there is $\delta>0$ such that $E[\gamma]\geq E[\gamma_0]$ for any $\gamma\in\mathcal{A}$ with $\|\gamma-\gamma_0\|_{H^2}\leq \delta$.

A straight line not touching $\psi$ is obviously a graph local minimizer in any case.
However, the straight line touching $\psi$ is not necessarily a local minimizer.
The existence of graph local minimizers touching $\psi$ is not trivial.

Moreover, it is shown that there are infinitely many self-intersecting local minimizers in $\mathcal{A}$ by using a kind of winding number.
For $\gamma\in\mathcal{A}$, the winding number $N_\gamma\in\mathbb{Z}$ is defined so that $2\pi N_\gamma$ is equal to the total curvature, or equivalently
\begin{eqnarray}
2\pi N_\gamma=\theta(1)-\theta(0), \nonumber
\end{eqnarray}
where $\theta:\overline{I}\to\mathbb{R}$ is a continuous representation of tangential angle (unique up to addition by a constant of $2\pi\mathbb{Z}$).
The winding number is obviously continuous with respect to the $C^1$-topology, and hence continuous with respect to the weak and strong $H^2$-topologies.
We denote by $\mathcal{A}_n\subset\mathcal{A}$ the set of all curves with $N_\gamma=n$.
Then, since $N_\gamma$ is discrete-valued and (weakly and strongly) continuous on $\mathcal{A}$, for any $n\in\mathbb{Z}$ the set $\mathcal{A}_n$ is open and closed in $\mathcal{A}$ with respect to both the weak and strong $H^2$-topologies.
Since $\mathcal{A}_n$ is weakly closed, in the same way as \ref{appendix1}, we can prove that for any $n\in\mathbb{Z}$ there is a minimizer of $E$ among $\mathcal{A}_n$.
Then, since $\mathcal{A}_n$ is strongly open, it turns out that such a minimizing curve is nothing but a local minimizer in the whole space $\mathcal{A}$.
Any curve with $N_\gamma\neq0$ has a self-intersection, and thus there are infinitely many self-intersecting (and overhanging) local minimizers in $\mathcal{A}$.
In the membrane setting $\mathcal{A}_*$, the above argument does not work since the winding number of any curve is zero, and thus the existence of overhanging local minimizers is nontrivial.

\subsection{Periodic boundary condition}

We finally give a brief remark on periodicity.
This paper assumes that admissible curves and minimizers have a same period, but the paper \cite{Ke16} proposes a numerical example of a global minimizer of a period several times a substrate period.
Hence, physically, a more natural assumption is that, if an original substrate has a period $\lambda$, then a minimizer $\gamma$ has the period $n\lambda$ for a positive integer $n$.
It is not easy to determine $n$ for a general case.
However, in terms of scale, our assumption would be formally justified.
In fact, the elasto-capillary length scale $\ell=\sqrt{C/\sigma_F}$ may be interpreted as the optimal bending scale of a minimizer (as in \cite{HuRoBi11,RoBi10}), and thus we would formally expect that a minimizing curve crosses over several periods of a substrate if and only if the scales $\ell$ and $\lambda$ balance each other out ($\ell\sim\lambda$), where $\lambda$ is the original substrate period (wavelength).
In our normalized setting, this balance is described as $\varepsilon\sim 1$.
The main concerns in this paper are the cases that $\varepsilon\gg1$ and $\varepsilon\ll1$ (even though our results give more precise conditions), and hence our periodic assumption would not be restrictive from this viewpoint.
The case that $\varepsilon\sim1$ is of course more interesting and challenging, but this paper is a first step and does not address the precise analysis.

\section{Conclusion}\label{sectconclusion}

We provided a first rigorous study on the graph representations of global minimizers (ground states) for the one-dimensional energy minimizing problem (\ref{minimizingproblem}), under the assumption that admissible curves and substrates have same periods.
We obtained ranges of some characteristic parameters ensuring the presence and absence of overhangs, respectively.
All the results are valid for both the filament setting (i.e., any self-intersection is allowed) and membrane setting (i.e., only self-contacts are allowed).
Our results do not give optimal conditions, but reveal which parameters are of importance in the mechanism of overhangs.
In particular, we found that the presence of overhangs crucially depends the precise relation between adhesiveness, normalized elasto-capillary length, ``deviation of a hall'', and ``sharpness of a mountain''.
This paper dealt with only special substrates to prove the presence of overhangs, so one future direction would be to understand the mechanism of overhangs for more general substrates.

\section*{Acknowledgments}

The author would like to thank Professor Olivier Pierre-Louis for his comments and discussion.
The author is grateful to the anonymous referees for their suggestions to improve the manuscript.
This work was supported by a Grant-in-Aid for JSPS Fellows 15J05166 and the Program for Leading Graduate Schools, MEXT, Japan.

\appendix

\section{Proof of existence}
\label{appendix1}

We confirm the existence theorem (Theorem \ref{thmexistence}) by a direct method in the calculus of variations.

\begin{proof}
	We first note that $\inf_\mathcal{A}E\leq\sigma_F$ by (\ref{energybound}) and the case $\inf_\mathcal{A}E=\sigma_F$ is trivial since a trivial straight line competitor is nothing but a minimizer.
	Thus we may assume that $\inf_\mathcal{A}E<\sigma_F$.
	
	Take a minimizing sequence $\{\gamma_n\}_n\subset\mathcal{A}$ such that
	\begin{eqnarray}
	\sigma_F>E[\gamma_n]\to\inf_\mathcal{A}E\ (\geq\sigma_B). \nonumber
	\end{eqnarray}
	Without loss of generality, we may assume that all the curves are of constant speed.
	In this case, the total energy of $\gamma_n$ is represented as
	\begin{eqnarray}
	E[\gamma_n]=\frac{C}{2L_{\gamma_n}^3}\int_I |\ddot{\gamma}_n(t)|^2dt + L_{\gamma_n}\int_I \sigma(\gamma_n(t))dt. \nonumber
	\end{eqnarray}
	
	Now we obtain the boundedness of $\{\gamma_n\}_n$ in $H^2(I;\mathbb{R}^2)$.
	Since $L_{\gamma_n}\sigma_B\leq E[\gamma_n]$, the sequence $\{L_{\gamma_n}\}_n$ is bounded.
	Thus, since $\gamma_n$ is of constant speed, the sequence of $\|\dot{\gamma}_n\|_2=L_{\gamma_n}$ is also bounded.
	Moreover, since $\frac{C}{2L_{\gamma_n}^3}\|\ddot{\gamma}_n\|_2^2\leq E[\gamma_n]$, the sequence of $\|\ddot{\gamma}_n\|_2$ is also bounded.
	Finally, since $E[\gamma_n]<\sigma_F$, we see that all the curves $\gamma_n$ must touch $\partial\Omega$.
	Combining this fact with the uniformly boundedness of length and the periodic boundary condition, we find that the sequence of $\|\gamma_n\|_\infty$ is bounded, and thus the sequence of $\|\gamma_n\|_2$ is also bounded.
	Therefore, the sequence $\{\gamma_n\}_n$ is bounded in $H^2(I;\mathbb{R}^2)$.
	
	Noting that $H^2(I;\mathbb{R}^2)$ is compactly embedded in $C^1(\bar{I};\mathbb{R}^2)$, there exists $\gamma\in H^2(I;\mathbb{R}^2)$ such that, up to a subsequence (not relabeled), $\gamma_n$ converges to $\gamma$ in $C^1$ and weakly in $H^2$.
	Notice that $\gamma\in\mathcal{A}$, the curve $\gamma$ is of constant speed, and $L_{\gamma_n}\to L_\gamma\geq1$.
	It only remains to prove $\liminf_{n\to\infty}E[\gamma_n]\geq E[\gamma]$.
	The lower semicontinuity of $\sigma$ and Fatou's lemma imply
	\begin{eqnarray}
	\liminf_{n\to\infty}\int_I \sigma(\gamma_n(t))dt \geq \int_I \sigma(\gamma(t))dt. \nonumber
	\end{eqnarray}
	Moreover, $\liminf_{n\to\infty}\|\ddot{\gamma}_n\|_2 \geq \|\ddot{\gamma}\|_2$ holds since $\ddot{\gamma}_n\to\ddot{\gamma}$ weakly in $L^2$.
	Noting the convergence of length, we obtain the lower semicontinuity of $E$.
	Consequently, the curve $\gamma$ is a minimizer.
\end{proof}

\end{document}